\documentclass{article}

\usepackage[latin1]{inputenc}
\usepackage{amsmath, amstext, amssymb, amsthm,epsfig}
\usepackage{float}
\usepackage{graphicx}
\usepackage{caption}
\usepackage{algorithm,algorithmic}
\usepackage{pdfsync}
\usepackage{authblk}
\usepackage{placeins,geometry}

\newtheorem{Thm}{Theorem}
\newtheorem{Cor}[Thm]{Corollary}
\newtheorem{Lem}[Thm]{Lemma}
\newtheorem{Prop}[Thm]{Proposition}
\newtheorem{Hyp}[Thm]{Hypothesis}

\newtheoremstyle{remexdef}{6pt}{13pt}{}{}{\bfseries}{.}{.5em}{\thmname{#1}~\thmnumber{#2}\thmnote{ (#3)}}
\theoremstyle{remexdef}
\newtheorem{Def}[Thm]{Definition}
\newtheorem{Rem}[Thm]{Remark}

\numberwithin{equation}{section}

\DeclareMathOperator{\Hess}{Hess}
\DeclareMathOperator{\Diag}{Diag}

\begin{document}

\newcommand{\Ss}{\mathbb{S}}
\newcommand{\N}{\mathbb{N}}
\newcommand{\Z}{\mathbb{Z}}
\newcommand{\Q}{\mathbb{Q}}
\newcommand{\R}{\mathbb{R}}
\newcommand{\C}{\mathbb{C}}
\newcommand{\LL}{\mathbb{L}}
\newcommand{\Prob}{\mathbb{P}}
\newcommand{\E}{\mathbb{E}}
\newcommand{\Filt}{\mathcal{F}}
\newcommand{\ind}{\boldsymbol{1}}
\renewcommand{\geq}{\geqslant}
\renewcommand{\leq}{\leqslant}
\newcommand{\id}{\mbox{Id}}
\newcommand{\txx}{\textcolor}

\newcommand{\lp}{\left(}
\newcommand{\rp}{\right)}
\newcommand{\lc}{\left[}
\newcommand{\rc}{\right]}
\newcommand{\lac}{\left\{}
\newcommand{\rac}{\right\}}
\newcommand{\lb}{\left|}
\newcommand{\rb}{\right|}

\renewcommand{\arraystretch}{1.3}


\renewcommand{\algorithmicrequire}{\textbf{Input:}}
\renewcommand{\algorithmicensure}{\textbf{Output:}}


\title{Monte Carlo methods for light propagation in biological tissues}
\author[1,4]{Laura Vinckenbosch}
\author[1,2]{C\'eline Lacaux}
\author[1,2]{Samy Tindel}
\author[3]{Magalie Thomassin}
\author[3]{Tiphaine Obara}

\affil[1]{Inria, BIGS, Villers-l\`es-Nancy, F-54600}
\affil[2]{Universit\'e de Lorraine, Institut \'Elie Cartan de Lorraine, UMR 7502, 
Vand\oe uvre-l\`es-Nancy, F-54506}
\affil[3]{Universit\'e de Lorraine, CRAN, UMR 7039, 9, avenue de la for\^et de Haye, Vand\oe uvre-l\`es-Nancy Cedex, 54516} 
\affil[4]{Universit\'e de Fribourg, D\'epartement de Math\'ematiques, chemin du Mus\'ee 23, CH-1700 Fribourg}

\date{\today}
\maketitle
\tableofcontents


\newpage

\begin{abstract}
Light propagation in turbid media is driven by the equation of radiative transfer. We give a formal probabilistic representation of its solution in the framework of  biological tissues 
and we implement algorithms based on Monte Carlo methods in order to estimate the quantity of light that is received by an homogeneous tissue when emitted by an optic fiber. A variance reduction method is studied and implemented, as well as a Markov chain Monte Carlo method based on the Metropolis-Hastings algorithm. The resulting estimating methods are then compared to the so-called Wang-Prahl (or Wang) method. 
Finally, the formal representation allows to derive a non-linear optimization algorithm close to Levenberg-Marquardt that is used for the estimation of the scattering and absorption coefficients of the tissue from measurements.


\end{abstract}

\section{Introduction}
The results presented in this article have initially been motivated by several research projects grounded on Photodynamic therapy (PDT), which is a type of phototherapy used for treating several diseases such as acne, bacterial infection, viruses and some cancers. The aim of this treatment is to kill pathological cells with a photosensitive drug that is absorbed by the target cells and that is then activated by light. For appropriate wavelength and power, the light beam makes the photosensitizer produce singlet oxygen at high doses and induces the apoptosis and necrosis of the malignant cells. See~\cite{Wilson1986, Wilson2008} for a review on  PDT.

The project that initiated this work focuses on an innovative application : the interstitial PDT for the treatment of high-grade brain tumors~\cite{Benachour2012,Bechet2014}. This strategy requires the installation of optical fibers to deliver light directly into the tumor tissue to be treated, while nanoparticles are used to carry the photosensitizer into the cancer cells.

Due to the complexity of interactions between physical, chemical and biological aspects and due to the high cost and the poor reproducibility of the experiments, mathematical and physical models must be developed to better control and understand PDT responses. In this new challenge, the two main questions to which these models should answer are:

\begin{enumerate}
\item What is the optimal shape, position and number of light sources in order to optimize the damage on malignant cells?
\item Is there a way to identify the physical parameters of the tissue which drive the light propagation?
\end{enumerate}

The  light propagation phenomenon involves three processes: absorption, emission and scattering that are described by the so-called equation of radiative transfer (ERT), see~\cite{Chandrasekhar1960}. In general, this equation does not admit any explicit solution, and its study relies on methods of approximation. One of them is its approximation by the diffusion equation and the use of finite elements methods to solve it numerically (see for example~\cite{Arridge1993}). An other approach, which appeared in the 1970s, is the simulation of particle-transport with Monte Carlo (MC) method (see~\cite{Bhan2007,Carter1975,ZhuLiu2013} and references therein). This method has been extended by several authors in order to deal with the special case of biological tissues and there is now a consensus in favor of the algorithm proposed by L.~Wang and S. L.~Jacques in~\cite{Wang1992_rep}, firstly described by S. A.~Prahl in \cite{Prahl1988} and S. A.~Prahl et al.~in \cite{Prahl1989}.  
This method is based on a probabilistic interpretation of the trajectory of a photon. It is widely used and there exist now turnkey softwares based on this method. However, this method is time consuming in 3D and the associated softwares lie inside some kind of black boxes. Due to a slight lack of formalism, it is difficult to speed it up while controlling the estimation error, or to adapt it to inhomogenous tissues such as infiltrating gliomas. Finally, even though there exist several methods in order to estimate the optical parameters of the tissue (see for example~\cite{KarlssonFredrikssonLarsson2012, FredrikssonLarssonStromberg2012, BargoPrahlGoodel2005, PalmerRamanujam2006}), one still misses formal representations that answer to the questions of identifiability.

In the current work, we wish to give a new point of view on simulation issues for ERT, starting from the very beginning. We  first derive a rigorous probabilistic representation of the solution to ERT in homogeneous tissues, which will help us to propose an alternative MC method to Wang's algorithm~\cite{Wang1992_rep}. Then we also propose a variance reduction method.

Interestingly enough, our formulation of the problem also allows us to design quite easily a Markov chain Monte Carlo (MCMC) method 
 based on Metropolis-Hastings algorithm. We have compared both MC and MCMC algorithms, and our simulation results show that the plain MC method is still superior in case of an homogeneous tissue. However, MCMC methods induce quick mutations, which paves the way to very promising algorithms in the inhomogenous case. 
Finally we handle the inverse problem (of crucial importance for practitioners), consisting in estimating the optical coefficients of the tissue according to a series of measurements. Towards this aim, we 
derive a probabilistic representation of the variation of the fluence rate with respect to the absorption and scattering coefficients. This leads us to  the implementation of a Levenberg-Marquardt type algorithm that gives an approximate solution to the inverse problem.

Our work should thus be seen as a complement to the standard algorithm described in~\cite{Wang1992_rep}. Focusing on a rigorous formulation, it opens the way to a thorough analysis of convergence, generalizations to MCMC type methods and a mathematical formulation of the inverse problem.

The paper is organized as follows. We derive the probabilistic representation of the solution to ERT in Section~\ref{sec:proba:representation}. In Sections~\ref{sec:montecarlo:approach} and~\ref{sec:MH}, we describe the MC and MCMC algorithms which are compared to Wang's algorithm in Section~\ref{sec:comparison}. Finally, the sensitivity of the measures with respect to the optical parameters of the medium, as well as their estimation are treated in Section~\ref{sec:inverse}.


\section{Probabilistic representation of the fluence rate}\label{sec:proba:representation}

\subsection{The radiative transfer equation}
 Let $D=\R^3$ be the set of positions in the {biological homogeneous} tissue and $\Ss^2$ be the unit sphere in $\R^3$. Let us denote the optical parameters of the tissue by $\mu_s>0$ for the \emph{scattering coefficient}, $\mu_a>0$ for the \emph{simple absorption coefficient} and  $\mu=\mu_s+\mu_a$ for the \emph{total absorption coefficient} (or \emph{attenuation coefficient}). Moreover, let us denote by $L_e(x,\omega)$ the \emph{emitted light} from $x$ in direction $\omega$ and by  $L(x,\omega)$  the \emph{quantity of light} at $x$ in the direction $\omega$. Then the equation of radiative transfer takes the following form (see e.g.~\cite{Prahl1988,Arvo1993}):

\begin{equation}\label{eq:ERT}
L(x,\omega)=L_i(x,\omega) + TL(x,\omega),\qquad x\in D,\, \omega\in\Ss^2,
\end{equation}
where $L_i(x,\omega)$ is the incident volume emittance 
and $T:L^{\infty}(\R^3 \times \Ss^2;\R)\to L^{\infty}(\R^3 \times \Ss^2;\R)$  is the linear operator defined on the Banach space of essentially bounded real-valued functions $\ell:\R^3 \times \Ss^2\to\R$,  given by
\begin{equation} \label{eq:def:operator:T}
T\ell(x,\omega)= \mu_s
\int_{\R_{+}} dr \, \exp(-\mu r) \int_{\Ss^2} d\sigma(\hat{\omega}) \, f(\omega,x-\omega r,\hat{\omega}) \, \ell(x-\omega r,\hat{\omega}), 
\end{equation}
with  $f$ the so-called \emph{bidirectional scattering distribution function} and $\sigma$  the uniform probability measure on the unit sphere $\mathbb{S}^2$.  
The incident volume emittance $L_i$ is also defined by applying a linear operator $T_i$ to $L_e$: 
\begin{equation}\label{eq:def-Li}
L_i  = T_{i} L_{e},
\quad\mbox{with}\quad
T_{i}\ell(x,\omega)
=
\int_{0}^{+\infty} \ell(x-r\omega,\omega) 
\exp(- \mu r) dr.
\end{equation} 

In the following, we will denote the \emph{albedo coefficient} by $\rho:=\frac{\mu_s}{\mu}<1$. Moreover, since we consider an homogeneous biological tissue, the scattering function is given by the so-called \emph{Henyey-Greenstein function}, see~\cite{Henyey1941}, that is
\begin{equation}\label{eq:def:fGH}
f(\omega,x,\hat{\omega})=f_{HG}(\omega,\hat{\omega})=\frac{1-g^2}{(1+g^2-2g\langle\omega,\hat{\omega}\rangle)^{3/2}}, \quad \omega,\hat{\omega}\in\Ss^2,\,\forall x\in D,
\end{equation}
where the constant $g\in [0,1)$ is the \emph{anisotropy factor} of the medium.  
The function $\hat{\omega}\mapsto f_{HG}(\omega,\hat{\omega})$ is a bounded and infinitely differentiable probability density function on $\Ss^2$ with respect to the uniform probability $\sigma$. It only depends on the angle $\theta$ between $\omega$ and $\hat{\omega}$ since $\langle \omega,\hat{\omega}\rangle=\cos(\theta)$. The greater the anisotropic parameter $g$, the more likely the large values of $\cos(\theta)$ and the less the scattering of the ray (see Fig.~\ref{fig:phase_function}). 

\begin{figure}[h!]
\begin{center}
\includegraphics[width=0.6\textwidth]{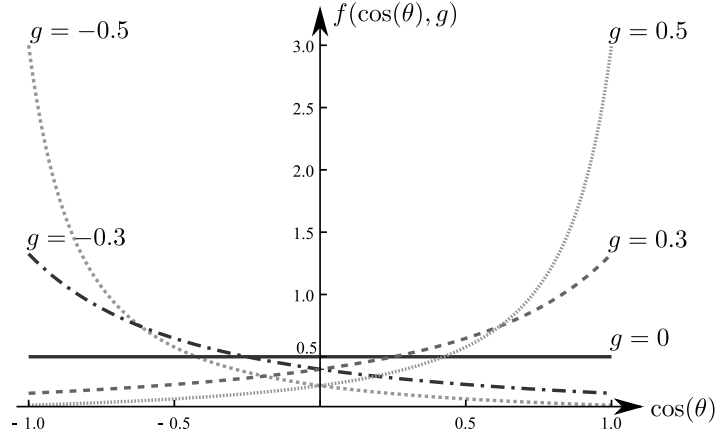}
\caption{Density function of $\cos(\theta)=\langle \omega,\hat{\omega}\rangle$ when $\hat{\omega}$ is scattered according to the Henyey-Greenstein function $f_{HG}(\omega,\cdot)$ for several values of the anisotropy parameter $g$.}\label{fig:phase_function}
\end{center}
\end{figure}


\subsection{Neumann series expansion of the solution}
In general, 
~(\ref{eq:ERT}) admits no analytical solution and a classical way to express its solution is to expand it 
in Neumann series. This method is based on the next classical and general result.

\begin{Thm}[\cite{Yoshida1980} p.69]\label{thm:Neumann}
Let $B$ be a Banach space equipped with a norm $\|\cdot\|$ and $A$ a linear operator on $B$. If  $\|A\|<1$, then the Neumann series $\sum_{n=0}^{\infty} A^n$ converges, the operator $\id - A$ is invertible and for any $x_0\in B$, the equation $x=Ax + x_0$ admits a unique solution given by
\begin{equation*}
x = (\id - A)^{-1} x_0 = \sum_{n=0}^{\infty} A^n x_0.
\end{equation*}
\end{Thm}

In order to apply Theorem \ref{thm:Neumann} in our context, let us now bound the norm of the operator $T$ defined above by \eqref{eq:def:operator:T}.
\begin{Lem}\label{lem:norm:T}
The operator $T$ defined in~\eqref{eq:def:operator:T}, with $f$ given by \eqref{eq:def:fGH},  
satisfies $\|T\|=\rho <1$, where we recall that we have set $\rho:=\frac{\mu_s}{\mu} <1$.
\end{Lem}

\begin{proof} 
Let $\ell\in L^{\infty}(\R^3 \times \Ss^2;\R)$. We have
\begin{eqnarray*}
|T\ell(x,\omega)|
&\leq& \mu_s
\int_{\R_{+}} dr \, \exp(-\mu r) \int_{\Ss^2  } d\sigma(\hat{\omega}) \, \left| f_{HG}(\omega,\hat{\omega}) \, \ell(x-\omega r,\hat{\omega}) \right|  \\
&\leq&\mu_s \|\ell\|_{\infty} \int_{\R_{+}} dr \, \exp(-\mu r)
= \frac{\mu_s}{\mu} \|\ell\|_{\infty},
\end{eqnarray*}
since $f_{HG}$ is a density function on $\Ss^2$. Thus, $\|T\|\leq\frac{\mu_s}{\mu}$ and since $T\bf1\equiv \frac{\mu_s}{\mu}$, we obtain $\|T\|=\frac{\mu_s}{\mu}$ and the proof is complete.
\end{proof}

As a corollary of the previous considerations, we are able to derive an analytic expansion for the solution to equation~\eqref{eq:ERT}:
\begin{Cor}\label{cor:ERT:Neumann}
If $L_e \in L^{\infty}(\R^3 \times \Ss^2;\R)$, then the radiative transfer equation \eqref{eq:ERT} with a phase function given by \eqref{eq:def:fGH} admits a unique solution $L$ in $L^{\infty}(\R^3  \times \Ss^2;\R)$. Moreover, $L$ can be decomposed as $L=\sum_{n=0}^{\infty} T^n L_i=\sum_{n=0}^{\infty} [T^n\circ T_{i}] \, L_{e}$ 
where  
$T^0\equiv\id$ and where for $n\ge 1$, 
 the linear operator $T^n\circ T_{i}$ on $L^{\infty}(\R^3  \times \Ss^2;\R)$ is given by:
\begin{multline}\label{eq:def:Tn}
[T^n\circ T_{i}] \, \ell 
(x,\omega_0)= 
\mu_s^n \int_{\R_{+}^{n+1}} dr_0\cdots dr_{n} \, \exp\left( -\mu \sum_{j=0}^{n} r_j\right)
\int_{(\Ss^{2})^{n}} d\sigma^{\otimes n}(\omega_1,\ldots,\omega_n) \\
\prod_{j=0}^{n-1}f_{HG}\left( \omega_{j},\omega_{j+1}\right) \,
\ell\left( x - \sum_{k=0}^{n} \omega_{k} r_{k} ,\omega_n\right) .
\end{multline}
\end{Cor}

\begin{proof} Assume that $L_e\in L^{\infty}(\R^3 \times \Ss^2;\R)$. It is readily checked from the definition \eqref{eq:def-Li} of $L_i$ that we also have $L_i\in L^{\infty}(\R^3 \times \Ss^2;\R)$. Indeed, 
$$
\left\|L_i\right\|_{\infty}\le \frac{\|L_e\|_{\infty}}{\mu}<+\infty.
$$
Hence, Theorem~\ref{thm:Neumann} and Lemma~\ref{lem:norm:T} provide the existence and uniqueness of the solution, as well as its expansion in Neumann series. Formula~\eqref{eq:def:Tn} is then found by induction.
\end{proof}

Our next step is now to recast representation \eqref{eq:def:Tn} into a probabilistic formula.

\subsection{Probabilistic representation}\label{sec:proba:repr}	
The Neumann expansion of $T$  enables us to express $L=\sum_{n=0}^{\infty}T^nL_i$ as an expectation. To this aim, we now introduce some notation. 
Let us define \begin{equation}\label{eq:def-A}
\mathcal{A}=\bigcup_{n=0}^{\infty}\mathcal{M}_{n},
\quad\mbox{ with }\quad
\mathcal{M}_n = \R_+^{n+1} \times \left(\Ss^2\right)^{n+1}.
\end{equation}
We denote by $( {\bf r},\boldsymbol{\omega})$ a generic element of $\mathcal{A}$ and by 
$({\bf r}_n,\boldsymbol{\omega}_n)$ a generic element of $\mathcal{M}_n$ for $n\in\mathbb{N}$ with ${\bf{r}}_n=( r_0,\ldots,r_n)$ and $\boldsymbol{\omega}_n=(\omega_0,\ldots,\omega_n)$. If $({\bf r},\boldsymbol{\omega})\in\mathcal{A}$, we set 
\begin{equation}
\label{rlength}
|{\bf r}|=\sum_{n=1}^{\infty} n \, \mathbf{1}_{\mathcal{M}_{n}}({\bf r},\boldsymbol{\omega})
\end{equation}
and call it \emph{size} or \emph{length} of the path.
For $n\in \mathbb{N}$, let\begin{equation}\label{eq:def:Gxn}
G_x^{(n)}({\bf r}_n,\boldsymbol{\omega}_n)=L_e\left( x - \sum_{k=0}^{n} \omega_{k} r_{k} ,\omega_n\right)
\end{equation}
be defined on $\mathcal{M}_n$, and let
\begin{equation}
\label{eq:def:Gx}
G_x({\bf r},\boldsymbol{\omega})= \sum_{n=0}^{\infty} G_x^{(n)}({\bf r}_n,\boldsymbol{\omega}_n) \, {\bf1}_{\mathcal{M}_n}({\bf r},\boldsymbol{\omega})
\end{equation}
be a function on $\mathcal{A}$. Let $Y=({\bf R},\boldsymbol{W})$ be a $\mathcal{A}$-valued random variable defined on a probability space $(\Omega,\mathcal{F},\mathbb{P})$, whose law $\nu$ is given by
\begin{equation}\label{eq:def:nu}
\nu(F) = \sum_{n=0}^{\infty} (1-\rho) \rho^{n} \nu_n\lp F\cap \mathcal{M}_n\rp,
\end{equation}
where we recall that $\rho=\frac{\mu_{s}}{\mu}$ and 
where $\nu_n$ is the probability measure on $\mathcal{M}_n$ defined by
\begin{equation}
\nu_n
(d{\bf r}_n, d\boldsymbol{\omega}_n)
= \mu^{n+1} \, e^{ -\mu \sum_{j=0}^{n} r_j} \prod_{j=0}^{n-1}\, f_{HG}\left( \omega_{j},\omega_{j+1}\right)
d{\bf r}_n
\, \sigma^{\otimes {(n+1)}}(d\boldsymbol{\omega}_n)\label{eq:def:nun}
\end{equation}
with  $\prod_{j=0}^{-1}a_j=1$ by convention.
Before we express $L$ as an expectation involving $G_x$ and $Y$, let us state some properties of $Y=({\bf R},\boldsymbol{W})$. 
\begin{Prop}\label{cor:mun:backward}
Let $n\in\N$ and let $\pi$ be a permutation of $\{0,\ldots,n\} $. 
Let us recall that $\nu$ defined by~\eqref{eq:def:nu} is the distribution of 
$Y=({\bf R},\boldsymbol{W})$. Then, conditionally on the event $\left\{ Y\in\mathcal{M}_n\right\}$, 
the distribution of the variable $(R_0,\ldots,R_n,W_0,\ldots,W_n)$ is the probability measure $\nu_n$ defined by \eqref{eq:def:nun}, which satisfies 
 \begin{equation}
 \label{symnun}
\nu_n(dr_0,\ldots,dr_n,d\omega_0,\ldots,d\omega_n)=
\nu_n(dr_{\pi(0)},\ldots,dr_{\pi(n)},d\omega_n,d\omega_{n-1},\ldots,d\omega_0).
\end{equation}
In other words, on the event 
$\left\{ Y\in\mathcal{M}_n\right\}$, 
the random variables 
 $(R_{\pi(0)},\ldots,R_{\pi(n)},W_{n},W_{n-1},\ldots,W_0)$ and $(R_0,\ldots,R_n,W_0,\ldots,W_n)$ have the same distribution $\nu_n$. 
\end{Prop}
\begin{proof}
By definition of $\nu=\mathcal{L}(Y)$ and since $\left(\mathcal{M}_p\right)_{p\in\mathbb{N}}$ is a collection of pairwise disjoint sets, it is straightforward that 
 $\nu_n$ is the distribution of the variable $(R_0,\ldots,R_n,W_0,\ldots,W_n)$ on the set $\left\{ Y\in\mathcal{M}_n\right\}$. 
Let us now emphasize that the phase function is symmetric, i.e.  $f_{HG}(\omega,\cdot)=f_{HG}(\cdot,\omega)$. 
Then replacing 
$f_{HG}(\omega_{j},\omega_{j+1})$ by $f_{HG}(\omega_{j+1},\omega_{j})$ and using the invariance of the definition of $\nu_n$  by  permutations of the variables $(r_{0},\ldots,r_{n})$, we obtain Equation \eqref{symnun}. 
\end{proof}
\begin{Cor}\label{cor:unif:angles} 
For all $n\in\N$, conditionally on the events $\{  Y\in \mathcal{M}_{n}
\}$ and $\{  Y\in \bigcup_{p\ge n}\mathcal{M}_{p}
\}$,    for any $j=0,\ldots,n$, 
the marginal distribution $\gamma_{W_{j}}$ of the direction $W_{j}$ is the uniform probability $\sigma$ on $\Ss^{2}$. In particular, $W_0$ is uniformly distributed on $\Ss^2$. 
\end{Cor}
\begin{proof} Let $n,j\in\N$ such that $j\le n$. 
 By proposition \ref{cor:mun:backward}, on the event $\left\{  Y\in \mathcal{M}_{n}\right\}$, the probability measure $\nu_n$ defined by~\eqref{eq:def:nun} is the  distribution of $({\bf R}_n,{\boldsymbol{W}}_n)$. Then, 
integrating the law $\nu_{n}$ with respect to all variables except  
$w_{j}$ and  using that $f_{HG}(\omega,\cdot)=f_{HG}(\cdot,\omega)$ is a density function on $\Ss^{2}$, one obtains  that   on the event $\left\{  Y\in \mathcal{M}_{n}\right\},$ $W_j$ is uniformly distributed on $\Ss^2$.  Since this also holds replacing $n $ by  any $p\ge n$ and since $\mathcal{M}_p$, $p\ge n,$ are pairwise disjoint sets, this implies that  on the event $\{  Y\in \bigcup_{p\ge n}\mathcal{M}_{n}\}$, $W_j$ is also uniformly distributed on $\Ss^2$.   
\end{proof}

\begin{Prop}\label{prop:expectation}
The series $\sum_{n=0}^{\infty} T^nL_i$ can be expressed as
\begin{equation}\label{eq:expectation:representation}
L(x,\omega)=\sum_{n=0}^{\infty} T^nL_i(x,\omega) = \frac{1}{\mu_a} \E \left[ G_x(Y) \mid W_{0}=\omega\right].
\end{equation}
\end{Prop}
\begin{proof} 
We shall relate $[T_{n}\circ T_{i}]L_{e}$ with the measure $\nu_{n}$ defined above, which will be sufficient for our purposes. To this aim, consider $\hat{\omega}_{0}\in\Ss^{2}$, and write a somehow more cumbersome version of formula~\eqref{eq:def:Tn} with $\ell=L_e$:
\begin{multline*}
[T^n\circ T_{i}] \, L_{e}(x,\hat{\omega}_0)= 
\mu_s^n \int_{\R_{+}^{n+1}} dr_0\cdots dr_{n} \, \exp\left( -\mu \sum_{j=0}^{n} r_j\right)
\int_{(\Ss^{2})^{n}} d\sigma^{\otimes n}(\omega_1,\ldots,\omega_n) \\
f_{HG}\left( \hat{\omega}_{0},\omega_{1}\right) 
\prod_{j=1}^{n-1}f_{HG}\left( \omega_{j},\omega_{j+1}\right) \,
L_e\left( x - \hat{\omega}_{0} r_{0} - \sum_{k=0}^{n} \omega_{k} r_{k} ,\omega_n\right) .
\end{multline*}
Noting that 
$\mu_{s}^{n} = \frac{(1-\rho) \, \rho^{n} \, \mu^{n+1}}{\mu_{a}}$,  
we thus get the following identity:
\begin{multline*}
[T^n\circ T_{i}] \, L_{e}(x,\hat{\omega}_0)= 
\frac{(1-\rho) \, \rho^{n} \, \mu^{n+1}}{\mu_{a}}
 \int_{\R_{+}^{n+1}} dr_0\cdots dr_{n} \, \exp\left( -\mu \sum_{j=0}^{n} r_j\right)
\int_{(\Ss^{2})^{n+1}} d\hat{\sigma}_{n,\hat{\omega}_0}(\omega_{0},\ldots,\omega_n) \\
\prod_{j=0}^{n-1}f_{HG}\left( \omega_{j},\omega_{j+1}\right) \,
L_e\left( x - \sum_{k=0}^{n} \omega_{k} r_{k} ,\omega_n\right) ,
\end{multline*}
where the measure $\hat{\sigma}_{n,\hat{\omega}_0}$ on $\Ss^{n+1}$ is given by 
$\hat{\sigma}_{n,\hat{\omega}_0}(d\omega_{0},\ldots,d\omega_{n})=\delta_{\hat{\omega}_{0}}(d\omega_{0}) \otimes \sigma^{\otimes n}(d\omega_{1},\ldots,d\omega_{n})$.

Finally set $\varphi_{n}({\bf r}_n, \boldsymbol{\omega}_n)= L_e( x - \sum_{k=0}^{n} \omega_{k} r_{k} ,\omega_n)$. Taking into account the identity above and \eqref{eq:def:nun}, we easily get
\begin{equation*}
[T^n\circ T_{i}] \, L_{e}(x,\hat{\omega}_0)= 
\frac{(1-\rho) \, \rho^{n}}{\mu_{a}}  \, 
\nu_{n}\lp \varphi_{n} | \,  \omega_{0}=\hat{\omega}_{0}\rp,
\end{equation*}
from which our claim is straightforward.  
\end{proof}

\begin{Rem}\label{rmk:def-Y}
In the following, we will call the random variable $Y=({\bf R},\boldsymbol{W})$ a \emph{ray}. Notice that it does not correspond exactly to a ray of light in the physics sense, since $Y$ has a finite length (though random) and since a given realization of $Y$ does not carry the information due to light absorption. Also notice that $Y$ owns a complete probabilistic description which allows to exactly simulate it (see Proposition~\ref{prop:simulation:and:RW} for simulation considerations). 
\end{Rem}


\subsection{Model for light propagation}

Observe that our formula \eqref{eq:expectation:representation} induces a Monte Carlo procedure to estimate $L(x,\omega)$ for each $(x,\omega)\in\mathbb{R}^{3}\times\mathbb{S}^{2}$ based on the simulation of independent copies of $Y$. Nevertheless this procedure is time consuming. Indeed, assuming that the light is only emitted by an optical fiber, many realizations $y$ of $Y$ lead to a null contribution in the estimation of $L(x,\omega)$.  Our aim is now to  accelerate our simulation by means of a coupling between random variables corresponding to different $(x,\omega)$. Towards this aim, we now focus on an averaged model for light propagation.

Let thus $V\subset\R^{3}$ be a cube whose center coincides with the origin. We discretize it into a partition of $K$ smaller cubes (voxels in the image processing terminology) $\{V_{k}, k=0,\ldots, K-1\}$, whose volume equals $h^{3}$, $h\in\R_{+}$ and such that the origin is the center of $V_{0}$.  Let us denote by $x_{k}$ the center of the voxel $V_{k}$. We work under the following simplified assumption for the form of the light source:
\begin{Hyp}\label{hyp:L-emitted}
We assume that the only emission of light in the domain $V$ comes from the optical fiber. Let $C^{2\alpha}\subset\Ss^{2}$ denote the cone with opening angle $2\alpha$, whose summit is placed at the origin and whose axis follows $-\vec{e}_3$.  The light source is defined by $S=\{(x,\omega): x\in V_{0},\,\omega\in C^{2\alpha} \}$. We assume that the emission of light satisfies
\begin{equation}\label{eq:def-L-emitted}
L_e(x,\omega)
=c \, \mathbf{1}_{V_{0}\times C^{2\alpha}}(x,\omega)
:=\left\{
                    \begin{array}{ll}
                      c, & \text{ if } (x,\omega)\in V_{0}\times C^{2\alpha}, \\
                      0, & \text{ otherwise},
                    \end{array}
                  \right.
\end{equation}
where $c>0$ is a given constant. 
\end{Hyp}
This model remains close to reality and it is possible to refine it by weighting the light directions of the source in order to stick better to the shape of the fiber. With Hypothesis~\ref{hyp:L-emitted} in mind, we are interested in estimating the \emph{fluence rate} at the center of the voxels $V_{k}$, $k\neq0$, that is the mean light intensity averaged in all directions 
\begin{equation}\label{eq:def-L-xk}
L(x_{k}):=\int_{\Ss^{2}} L(x_{k}, \omega_{0}) \, \sigma(d\omega_{0}).
\end{equation}
This quantity admits a nice probabilistic representation.

\begin{Prop}
Let $k\in \{0,\ldots, K-1\}$ and let $Y=\left({\boldsymbol R},\boldsymbol{W}\right)$ be a random variable with distribution $\nu$ defined by \eqref{eq:def:nu}. Then, the fluence rate $L(x_{k})$ at the center $x_k$ of the voxel $V_k$, which is defined by~\eqref{eq:def-L-xk}, can also be expressed as
\begin{equation}\label{eq:luminance:proba:repres}
L(x_{k}) =
\frac{c}{\mu_a}\Prob\left(x_{k}-\sum_{j=0}^{\left|{\boldsymbol{R}}\right|}R_{j}W_j\in V_{0} ,\,W_{\left|{\boldsymbol{R}}\right|}\in C^{2\alpha}\right)
\end{equation}
where we recall that  the length $\left|{\boldsymbol{R}}\right|$ of the ray $Y$ is defined by \eqref{rlength}. 
\end{Prop}

\begin{proof}
Invoking Proposition~\ref{prop:expectation} and by definition of $L$,  we get 
\begin{equation*}
L(x_{k})
= \frac{1}{\mu_a}\int_{\Ss^2}  \E\left[G_{x_{k}}(Y)|W_0=\omega_0\right]\, \sigma(d\omega_{0})
\end{equation*}
 where $G_{x_k}$ is defined by \eqref{eq:def:Gx}. Since by Corollary \ref{cor:unif:angles}, $\sigma$ is the distribution of $W_0$, the previous equation can be written as
$$
 L(x_{k})
=\frac{1}{\mu_a}\,\E\left[G_{x_{k}}(Y)\right].
$$ 
 Then using equations \eqref{eq:def:Gxn} and \eqref{rlength}, we get 
 $$
 G_{x_k} (Y)=L_e\left(x_k-\sum_{j=0}^{\left|{\boldsymbol{R}}\right|}R_{j}W_j,\,W_{\left|{\boldsymbol{R}}\right|}\right).
 $$
Now applying Hypothesis \ref{hyp:L-emitted}, the random variable $G_{x_{k}}(Y)$ can be expressed as 
$$
G_{x_{k}}(Y)=c \, \mathbf{1}_{\lac x_{k}-\sum_{j=0}^{\left|{\boldsymbol{R}}\right|}R_{j}W_j\in V_{0} ,\,W_{\lb\boldsymbol{R}\rb}\in C^{2\alpha}\rac},
$$
which finishes our proof.
\end{proof}


Now, instead of seeing $Y$ as a ray starting at $x_{k}$ which possibly hits the light source $V_{0}\times C^{2\alpha}$, we can imagine that it starts at the center of the light source and it possibly hits the voxel $V_{k}$ in any direction.  This possibility stems from the invariance of $\nu$ stated in Proposition \ref{cor:mun:backward}, and is exploited in the next result.

\begin{Prop}
\label{prop:simulation:and:RW} 
For any $0\leq k\leq K-1$ we have :  
\begin{equation}\label{eq:Lxk:RW}
L(x_{k})=\frac{c( 1-\cos\alpha)}{2\mu_a}\,\Prob\left(S_{N}\in V_{k}\right),
\end{equation}
where 
$N\sim {\rm NB}(1,\rho)$ is a negative binomial random variable with parameter $(1,\rho)$, that is a random variable whose law is given by $\mathbb{P}(N=n)=(1-\rho)\rho^{n}$ for  all $ n\in \N,$ and where for $n\in\N$
$$ S_n=\sum_{i=0}^n R_i W_i$$
with  $(R_i)_{i\ge 0}$ and $(W_i)_{i\ge 0}$ satisfying the following assertions:
\begin{itemize}
\item $(R_i)_{i\ge 0}$ is a sequence of independent identically distributed (i.i.d.) exponentially random variables of parameter $\mu$ (i.e such that $\mathbb{E}(R_i)=\mu^{-1}$).
\item $W_0$ is uniformly distributed on the cone $C^{2\alpha}$.
\item for any $i\ge 1$, the conditional distribution of $W_i$ given $(W_0,\ldots,W_{i-1})=(\omega_0,\ldots,\omega_{i-1})$ is $f_{HG}(\omega_{i-1},\omega_i) \sigma(d\omega_i)$.
\item $N$, $(R_i)_{i\geq 0}$ and $(W_i)_{i\geq 0}$ are independent. 
\end{itemize}
\end{Prop}
\begin{proof} Let $k\in \lac 0,\ldots, K-1\rac$. Notice that, if $V_{k}+x$ denotes the translation of the voxel $V_{k}$ by the vector $x\in\R^{3}$, then it is clear that $V_{0}+x_{k}=V_{k}$. 
 Therefore,  we can rewrite~\eqref{eq:luminance:proba:repres} as
\begin{equation}
L(x_k)=\frac{c}{\mu_a} \mathbb{P}\left(\sum_{j=0}^{\left|{\boldsymbol R}\right|}R_{j}W_j \in V_k,
W_{\left|{\boldsymbol R}\right|}\in C^{2\alpha}
\right)\,,\label{eq:fluence:as:proba}
\end{equation}
where the distribution of  $Y=({\boldsymbol R}, {\boldsymbol{W}})$ is the probability measure $\nu$ defined by \eqref{eq:def:nu}. 
Then, 
$$L(x_k)=\frac{c}{\mu_a} \sum_{n=0}^{+\infty}(1-\rho) \rho^{n} \mathbb{P}\left(\sum_{j=0}^{n} R_{j}W_j \in V_k, W_n\in C^{2\alpha}\right)$$
where the distribution of $(R_0,\ldots, R_n, W_0,\ldots,W_n)$  is the probability $\nu_n$ defined by \eqref{eq:def:nun}. Therefore, applying Proposition \ref{cor:mun:backward}, we get
$$
L(x_k)=\frac{c}{\mu_a} \sum_{n=0}^{+\infty}(1-\rho) \rho^{n} \mathbb{P}\left(\sum_{j=0}^{n} R_{j}W_j \in V_k, W_0\in C^{2\alpha}\right).
$$
By definition of $\nu_n$, we easily see that 
$$\mathbb{P}\left(\sum_{i=0}^{n} R_{i}W_i \in V_k,\, W_0\in C^{2\alpha}\right)=\sigma(C^{2\alpha}) \mathbb{P}\left(\sum_{i=0}^n R_i'W_i' \in V_k\right), 
$$
where $(R_0',\ldots,R_n', W_0',\ldots,W_n')$ is a random variable with distribution 
$
\nu_n'
$, defined by replacing in~\eqref{eq:def:nun} the uniform probability $\sigma(d\omega_0)$ on the sphere $\Ss^2$  by the uniform probability  $\ind_{\{\omega_0'\in C^{2\alpha}\}}\sigma(d\omega_0')/\sigma(C^{2\alpha})$ on the cone $C^{2\alpha}$. By definition of $\nu_n'$, this  means that 
$$
\mathbb{P}\left(\sum_{i=0}^{n} R_{i}W_i \in V_k,\, W_0\in C^{2\alpha}\right)=\sigma(C^{2\alpha}) \mathbb{P}(S_n \in V_k)
$$
where $S=(S_p)_{p\ge 0}$ is the random walk defined in the statement of the proposition.  
Thanks to the fact that $N$ is a ${\rm NB}(1,\rho)$  random variable independent of the random walk $S$, we easily get relation~\eqref{eq:Lxk:RW}. 
\end{proof}

\begin{Rem} 
The random variables $N,R_{i}$ and $W_{0}$ in Proposition \ref{prop:simulation:and:RW} are simulated in a straightforward way.
The simulation of the sequence $(W_{i})_{i\geq1}$ is obtained as follows. The direction $W_{i}$ of the $i$-th step of the random walk is sampled relatively to the direction $W_{i-1}$. We sample the spherical angles $(\Theta_{i},\Phi_{i})$ \emph{between} the two directions according to the Henyey-Greenstein phase function. Namely, $\cos\lp\Theta_{i}\rp=\langle W_{i-1}, W_{i}\rangle$ owns the following invertible cumulative distribution function:
$$F^{-1}(y)=\frac{1}{2g}\,\lp 1+g^{2}-\lp\frac{1-g^{2}}{1-g+2gy}\rp^{2}  \rp, \qquad y\in[0,1]$$
and $\Phi_{i}$, the azimuth angle of $W_{i}$ in the frame linked to $W_{i-1}$, is uniformly distributed on $[0,2\pi]$, see~\eqref{eq:def:fGH}. To recover the cartesian coordinates of the directions, we inductively apply appropriate changes of frame. The corresponding formulas can be found in~\cite[p. 37]{Prahl1988}.
\end{Rem}

\begin{Rem}\label{rem:simultaneous:estimates}
Notice that $S_{N}$ does not depend on the voxel $x_{k}$ under consideration. This permits to use a single sample of realizations of this random variable in order to estimate the right hand side~of~\eqref{eq:Lxk:RW} for all $k\in\lac0,\ldots, K-1 \rac$ simultaneously. {We call this improvement coupling, meaning that the random variables related to the Monte Carlo evaluations at different voxels are completely correlated.}
\end{Rem}

%


\section{Monte Carlo approach with variance reduction}\label{sec:montecarlo:approach}

In the last section, we have derived a probabilistic representation of $L(x_{k})$ for every voxel $V_{k}$ by means of the arrival position of a random walk $(S_{n})_{n\geq1}$ stopped at a negative binomial time. This classically means that $L(x_{k})$ can be approximated by MC methods. We first derive the expression of the approximate fluence rate by means of the basic MC method and then describe the variance reduction method that we implemented.

\begin{Prop} Let  us consider a random walk $S=(S_n)_{n\ge 0}$ and a negative binomial random variable $N$ as defined in Proposition \ref{prop:simulation:and:RW}. Let $(S^{i}, N_i)_{1\le i \le M}$ be  $M$ independent copies of $(S,N)$. Then, for  $k=0,\ldots,K-1$, 
\begin{equation}\label{eq:def:thetak:pk}
\widehat{L}_{\text{\emph{MC}}}(x_{k}):=\frac{c \left(1-\cos\alpha\right)}{2\mu_a M}\,\sum_{i=1}^M \ind_{\left\{S^{i}_{N_i}\in V_{k}\right\}}
\end{equation}
is an unbiased and strongly consistent estimator of $L(x_k)$.
\end{Prop}
\begin{proof}
This statement follows simply from the discussion of the previous section and the law of large numbers (LNN).
\end{proof}

In addition to this proposition, let us highlight the fact that central limit theorem provides confidence intervals for the estimators $\widehat{L}_{\text{MC}}(x_{k})$.  Furthermore, owing to Remark~\ref{rem:simultaneous:estimates},  the family $(S^{i}, N_i)_{1\le i \le M}$ enables to estimate the fluence rate $L(x_k)$ for all $k=0,\dots,K-1$ at once. \\

The reader should be aware of the fact that the quantity $\rho$ is large in general in biological tissues, which means that the size of the ray, given by $N\sim{\rm NB}(1,\rho)$, will often be large. Typical values of the parameters are provided in Table~\ref{table:parameters:value}. Therefore, sampling a ray is relatively time consuming and it is necessary to improve the basic Monte Carlo algorithm in order to reduce the variance of the estimates. 

\begin{table}[H]
\begin{center}
  \begin{tabular}{|l|c|c|c|c|c|c|}
    \hline
    &$\mu_{s}$ & $\mu_{a}$ & $\mu$ & $\rho$ & $g$ \\ 
    \hline
   healthy tissue &$280\text{ cm}^{-1}$ & $0.57\text{ cm}^{-1}$ & $280.57\text{ cm}^{-1}$ & $0.998$ & $0.9$ \\ 
    \hline
  tumor &$73\text{ cm}^{-1}$ & $1.39\text{ cm}^{-1}$ & $74.39\text{ cm}^{-1}$ & $0.981$ & $0.9$ \\ 
    \hline
    \end{tabular}
  \caption{Values given by~\cite{AngellPetersen2007} for the optical parameters in the rat brain for a wavelength $\lambda=632$ nm (red light).}\label{table:parameters:value}
\end{center}
\end{table}

Furthermore, because of the formulation~\eqref{eq:fluence:as:proba}, only the last point of each whole ray is used in the estimation. It is however possible to take into account more points of the rays and still have unbiased estimators. Finally, the angular symmetry of the problem, allows us to replicate observed rays by applying rotation. We took these two considerations into account and named the resulting method \textbf{Monte Carlo with some points (MC-SOME)}. The idea is to firstly draw some random walks which share the same initial direction and to pick a given number of points of each walk. Then, we apply rotations to that set of points with respect to different initial directions. We finally count the number of points in each voxel. This artificially increase the size of the samples and thus reduce the variance of our estimation of $L(x_k)$. Specifically, the resulting estimator is given by :

\begin{Def}[MC-SOME]\label{def:MCSOME}
Let $M, M_{\text{points}}, M_{\text{rot}}\in\N^{*}$ be the parameters of the method. Let us assume that the following assertions hold:
\begin{itemize}
\item $(W_{0}^{j})_{1\le j\le M_{\text{rot}}}$ are i.i.d.~copies of $W_{0}\sim\mathcal{U}(C^{2\alpha})$.
\item $N^{\ell}_{i}$, $1\leq i\leq M$, $1\leq\ell\leq M_{\text{points}},$ are i.i.d. copies of a negative binomial random variable with parameter~$(1,\rho)$. 
\item  
$\lp S^{i}_{n}\rp_{n\ge 0}$, $1\le i\le M$, are i.i.d.~copies of the random walk $S$ defined in Proposition~\ref{prop:simulation:and:RW}, all sharing the same initial direction $W_{0}^{1}$.
\item The sequences $\lp N^{\ell}_{i}\rp_{i,\ell}$ and $ \lp W_{0}^{j}, S^{i}_{n}\rp_{i,j}$ are independent. 
\end{itemize}

\noindent
Let $S^{i,j}_{N^\ell_i}$ denotes the $N^\ell_i$-th point of the $i$-th random walk $\lp S^{i}_{n}\rp_{n\ge 0}$ after a rotation corresponding to the $j$-th initial direction $W_{0}=W_{0}^{j}$.

\noindent
Then, for $k\in\{0,\ldots,K-1\}$, the MC-SOME estimate of $L(x_k)$ is defined by
\begin{equation}\label{eq:def:estima:MCSome}
\widehat{L}_{\text{MC-SOME}}(x_k) =\frac{c \left(1-\cos\alpha\right)}{2\mu_a M_{\text{rot}}M_{\text{points}}M}\sum_{i=1}^{M}\sum_{\ell=1}^{M_{\text{points}}}\sum_{j=1}^{M_{\text{rot}}}
\ind_{\left\{S^{i,j}_{N^\ell_i}\in V_{k}\right\}}\,.
\end{equation}
\end{Def}
This estimator is unbiased and strongly consistent. Its construction is illustrated in Fig.~\ref{fig:path:MCSOME} and it will be compared to Wang's algorithm estimator in Section~\ref{sec:comparison}.

\begin{figure}[H]
\begin{center}
\includegraphics[width=0.85\textwidth]{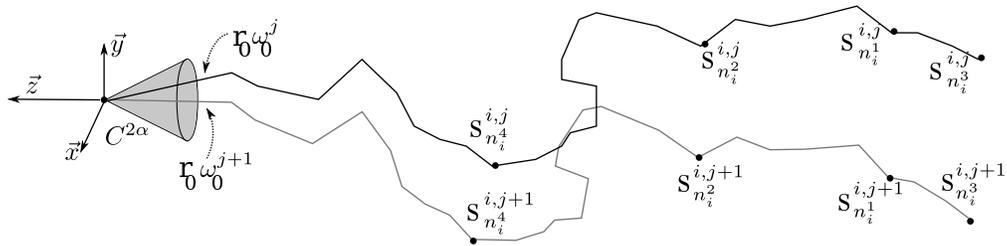}
\caption{Description of  MC-SOME method. In this example, the grey path is a rotation of the black one with respect to its initial direction $\omega_{0}^{j+1}$ and $M=4$.}\label{fig:path:MCSOME}
\end{center}
\end{figure}
\begin{Rem}
Choosing only some points $\lp S_{N^\ell}, 1\leq\ell\leq M_{\text{points}}\rp$ of the path instead of all $\lp S_i, 1\leq i\leq n\rp$ provides a more efficient estimation. Indeed, it would take a lot of time to run over all voxels $V_k$ in order to evaluate the indicator functions $\ind_{\left\{S_i\in V_{k}\right\}}$. Moreover, the information brought by close points is in a sense redundant. Choosing the points according to a negative binomial law maintain the estimator unbiased, while speeding up the estimation.
\end{Rem}


\section{A Metropolis-Hastings algorithm for light pro\-pagation}\label{sec:MH}
Inspired by results in computer graphics, see~\cite{Jensen2003, Shirley2008,Veach1997_thesis}, we implemented a  Metropolis Hastings algorithm which is a \emph{Markov chain Monte Carlo method} (MCMC) by random walk~\cite{Hastings1970, Roberts2004}. We shall first discuss general principles and then practical implementation issues.

For simplicity reasons,  by slightly abusing the notations, we identify the stopped random walk $S=(S_{n})_{0\leq n\leq N}$ of Proposition~\ref{prop:simulation:and:RW} with the ray $(R_0,\ldots, R_N,W_0,\ldots,W_N)$ which defines it. The law of this ray, which will still be denoted by $\nu$,  is given by replacing in~\eqref{eq:def:nu} the uniform measure $\sigma(d\omega_0)$ on the sphere by the uniform measure on the cone $C^{2\alpha}$.

A realization of the walk $S$ stopped at $N=n$ will be indifferently referred to as $(S_{0},\ldots, S_{n})$, as $\lp r_{0},\ldots,r_{n},\omega_0, \ldots, \omega_{n}\rp$ or as $\lp r_{0},\theta_0,\varphi_0, \ldots,r_{n},\theta_{n},\varphi_{n}\rp$ where $(\theta_{0}, \varphi_{0})$ are the spherical coordinates of $\omega_{0}$ and for $1\leq i\leq n$, $\cos(\theta_{i})=\langle \omega_{i-1},\omega_{i}\rangle$ and $\varphi_{i}$ is the azimuth angle of $\omega_{i}$ in the frame linked to $\omega_{i-1}$.

\subsection{General principle}
For a given $\omega_{0}\in C^{2\alpha}$ and for $0\leq k\leq K-1$, we are willing to estimate the conditional probability $\mathbb{P}\left(S_N\in V_k \mid W_{0}=\omega_{0}\right)$ by generating a Markov chain whose steady-state measure is the conditional distribution $\nu |_{W_{0}=\omega_{0}}$ and by applying LNN for ergodic Markov chains.  We then combined this estimation with the classical LNN sampling the initial direction $W_{0}$ in $C^{2\alpha}$ to obtain an estimate of $L(x_{k})$ viewed as in~\eqref{eq:Lxk:RW}.

An overview of the MCMC dynamics in this context is the following:  Let $\omega_{0}\in C^{2\alpha}$ be a fixed initial direction. The Markov chain starts at time $t=1$ in the state $S(1)\in \mathcal{A}$ with $W_{0}=\omega_{0}$. At each time $t\in\N^{*}$, a move (mutation) is proposed from the current state $S(t)$ to the state $S^{\prime}(t)$ according to a \emph{proposal density} $q(S(t),\cdot)$ and such that the initial direction of $S^{\prime}(t)$ is still $\omega_{0}$. The chain then jumps to $S^{\prime}(t)$ with \emph{acceptance probability} $\alpha(S(t),S^{\prime}(t))$ or stays in $S(t)$ with probability  $1-\alpha(S(t),S^{\prime}(t))$. This is described in pseudo-code in Algorithm~\ref{algo:MH}. The MCMC simulation generates a Markov chain $\lac S(t); t \geq 1\rac$ on the space of rays $\mathcal{A}$ whose steady-state measure is the desired distribution $\nu|_{W_{0}=\omega_{0}}$ (see~\cite[\S  2.3.1]{Tierney1994}).  

\begin{algorithm}[H]
\caption{Metropolis-Hastings algorithm for light propagation}\label{algo:MH}
\begin{algorithmic}
\STATE \textbf{Initialization:} 
\STATE draw $\omega_{0}$ uniformly on $\mathcal{C}^{2\alpha}$,
\STATE  draw $S(1)$ according to $\nu|_{W_{0}=\omega_{0}}$
\FOR {$t=1$ to $T-1$}
   \STATE  $S^{\prime}(t) \sim q(S(t),\cdot)$
   \STATE  $\alpha( S(t), S^{\prime}(t)) \leftarrow \min \lac 1, \dfrac{\nu|_{W_{0}=\omega_{0}}\lp S^{\prime}(t)\rp\, q(S^{\prime}(t) , S(t))}{\nu|_{W_{0}=\omega_{0}}\lp S(t)\rp\, q(S(t) , S^{\prime}(t))} \rac$
   \IF {Rand()$<\alpha( S(t), S^{\prime}(t))$}
   \STATE $S(t+1) \leftarrow S^{\prime}(t)$
   \ELSE
   \STATE $S(t+1)\leftarrow S(t)$
   \ENDIF
\ENDFOR
\end{algorithmic}\end{algorithm}

If the ray $S(t)=(S_{0}(t),\ldots,S_{N_{t}}(t))$ denotes the position of the chain at a given time $t$, then, for $0\leq k\leq K-1$, 
\begin{equation}\label{eq:estim:pk:conditional}
\lim_{T\to \infty}\,\frac{1}{T} \sum_{t=1}^{T} \ind_{\lac S_{N_{t}}(t)\in V_{k}\rac}=\mathbb{P}\left(S_{N}\in V_k \mid W_{0}=\omega_{0}\right) \qquad \text{ almost surely,}
\end{equation}
on condition that the chain $\lac S(t); t \geq 1\rac$ is Harris positive with respect to $\nu|_{W_{0}=\omega_{0}}$. Indeed, this statement relies on LLN for Harris recurrent ergodic chains, see~\cite[Theorem 17.0.1]{Meyn2009}.

We can then sample the law of $W_{0}$ to recover an estimate of $\mathbb{P}\left(S_{N}\in V_k \right)$. Let $M_{\text{rot}}\in\N^{*}$ and let $(\omega_{0}^{1}, \ldots, \omega_{0}^{M_{\text{rot}}})$ be a sample of i.i.d. initial directions drawn according to $\mathcal{U}\lp C^{2\alpha}\rp$. For $i=1,\ldots, M_{\text{rot}}$ and $t=1, \ldots, T$, let $S^{(i)}(t)$ denote the rotation of the random walk $S(t)$ with respect to the initial direction $\omega_{0}^{i}$, see Fig.~\ref{fig:path:MCSOME}. For $k\in\{0,\ldots,K-1\}$, our Metropolis-Hastings estimator of $L(x_k)$ is defined by
\begin{equation}\label{eq:def:estima:MH}
\widehat{L}_{\text{MH}}(x_k) =\frac{ c \left(1-\cos\alpha\right)}{2\mu_a M_{\text{rot}}\,T}\sum_{i=1}^{M_{\text{rot}}}\sum_{t=1}^{T}
\ind_{\left\{S^{(i)}_{N_t}(t)\in V_{k}\right\}}\,.
\end{equation}
This estimator is strongly consistent on condition that the proposal density $q(\cdot,\cdot)$ of Algorithm~\ref{algo:MH} provides a Harris recurrent chain.


\subsection{Mutation strategy}
One of the delicate issues in the implementation of MCMC methods is the choice of a convenient proposal density $q$. In our case, we have tested a mixture of mutations of two types for the Metropolis-Hastings algorithm: \emph{rotation-translation} and \emph{deletion-addition}. In order to describe them, let us introduce some notations. The method uses a \emph{perturbed phase function} $f^{\varepsilon g}_{HG}$, $\varepsilon\in [-1,1]$ which has an anisotropic coefficient $\varepsilon g$  instead of $g$. It also uses two coprime integers $1\leq j< J$, which denote respectively the small and the big size of a length change. 

\begin{Rem}
The need for these two numbers to be coprime will be made clearer in the proof of Proposition~\ref{prop:convergence:MH}, as it will ensure that the mutations produce paths of arbitrary length.
\end{Rem}

\begin{Def}[Mutation rule]\label{def:mutations:rule}
Let us assume that, at time $t\in\N^{*}$, the current ray is given by  $S(t)=(r_0,\theta_{0},\varphi_{0},\ldots,r_{n_{t}},\theta_{n_{t}},\varphi_{n_{t}} )$. Our proposition for the next move from $S(t)$ to $S(t+1)$ is the following.

\smallskip

\noindent
\emph{(i)}
With probability $\frac{1}{2}$, the mutation is of type \emph{deletion-addition}. Draw $\Delta(n_{t})$ according to the following law that depends on the size of the current ray
$$ \Delta(n)\sim\begin{cases} 
\mathcal{U}\lp\lac -J,-j,j,J\rac\rp, &\text{if } n\geq J,\\
\mathcal{U}\lp\lac -j,j,J\rac\rp, &\text{if } j\leq n< J,\\
\mathcal{U}\lp\lac j,J\rac\rp, &\text{if } 0\leq n < j.\\
\end{cases}$$

\begin{itemize}
\item If $\Delta(n_{t})<0$, then delete the last $\lb \Delta(n_{t}) \rb$ edges of $S(t)$. The proposed path is 
$$S^{\prime}(t)=\lp r_0,\theta_{0},\varphi_{0}, \ldots,r_{n_{t}-\Delta(n_{t})},\theta_{n_{t}-\Delta(n_{t})},\varphi_{n_{t}-\Delta(n_{t})}\rp.$$

\item If $\Delta(n_{t})>0$, then add $\Delta(n_{t})$ new edges at the end of $S(t)$:
  	\begin{itemize}
     	\item[-] Draw $\lp r_{n_{t}+1}^{\text{new}},\ldots, r_{n_{t}+\Delta(n_{t})}^{\text{new}}\rp$  i.i.d. according to {$\mathcal{E}(\mu)$}.	
	\item[-] Draw $\lp \theta_{n_{t}+1}^{\text{new}},\ldots, \theta_{n_{t}+\Delta(n_{t})}^{\text{new}}\rp$  i.i.d. according to $f^{\varepsilon g}_{HG}$.
	\item[-] Draw $\lp \varphi_{n_{t}+1}^{\text{new}},\ldots, \varphi_{n_{t}+\Delta(n_{t})}^{\text{new}}\rp$ i.i.d. uniformly on $[0,2\pi]$.
	\end{itemize}
The proposed path is 
$$S^{\prime}(t)=\lp r_0,\theta_{0},\varphi_{0},\ldots, r_{n_{t}},\theta_{n_{t}},\varphi_{n_{t}}, r_{n_{t}+1}^{\text{new}}, \theta_{n_{t}+1}^{\text{new}}, \varphi_{n_{t}+1}^{\text{new}}, \ldots, r_{n_{t}+\Delta(n_{t})}^{\text{new}}, \theta_{n_{t}+\Delta(n_{t})}^{\text{new}},\varphi_{n_{t}+\Delta(n_{t})}^{\text{new}}\rp.$$
\end{itemize}

\smallskip

\noindent
\emph{(ii)}
With probability $\frac{1}{2}$, the mutation is of type \emph{rotation-translation}. Choose an index $i$ uniformly over $\lac 0,\ldots, n_{t}\rac$.

\begin{itemize}
\item If $i\neq0$, then make a rotation of the path at the $i$-th edge:
	\begin{itemize}
	\item[-] Draw a new angle $\theta_{i}^{\text{new}}$ according to $f^{\varepsilon g}_{HG}$.
	\item[-] Draw a new angle $\varphi_{i}^{\text{new}}$ uniformly on $[0,2\pi]$.
	\end{itemize}
The proposed path is (see Fig.~\ref{fig:mutation:rotation})
$$S^{\prime}(t)=\lp r_0,\theta_{0},\varphi_{0}, \ldots, r_{i}, \theta^{\text{new}}_{i},\varphi^{\text{new}}_{i},\ldots,r_{n_{t}},\theta_{n_{t}},\varphi_{n_{t}} \rp.$$

\item If $i=0$, then translate from the initial edge:
	\begin{itemize}
	\item[-] Draw a new edge length $r_{0}^{\text{new}}$ according to $\mathcal{E}(\mu)$.
	\end{itemize}
The proposed path is 
$$S^{\prime}(t)=\lp r_{0}^{\text{new}}, \theta_{0},\varphi_{0},  r_1, \theta_{1}, \varphi_{1},\ldots,r_{n_{t}},\theta_{n_{t}},\varphi_{n_{t}}\rp.$$
\end{itemize}
\end{Def}

\begin{Rem}
The initial direction is fixed, thus mutations of $(\theta_0,\varphi_0)$ are forbidden. However, the length $r_0$ has to change from time to time. This is ensured by the translations and this is why the initial edge need to be considered separately in the mutation rule.
\end{Rem}

\begin{figure}[h]
\begin{center}
\includegraphics[width=0.55\textwidth]{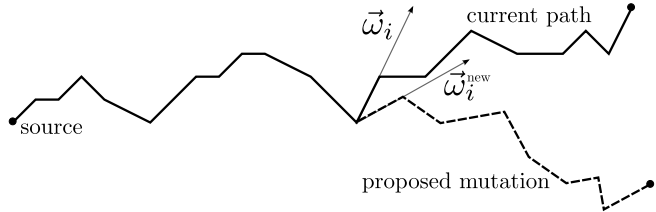}\hspace{8ex}
\caption{Metropolis-Hastings algorithm. Example of a mutation by rotation.}\label{fig:mutation:rotation}
\end{center}
\end{figure}

Let us now compute the proposal density $q(S,\cdot)$ of this mutation rule.  For $m\in\N^{*}$, let 
$$\zeta(m)=\begin{cases}
\frac{1}{4}\,, &\text{if } m \geq J,\\
\frac{1}{3}\,, &\text{if } j\leq m< J,\\
\frac{1}{2}\,, &\text{if } 0\leq m < j.\\
\end{cases}$$
Assume that $S^{\prime}$ is a mutation of $S$, denote by $i$ the first index where there is a difference between them, denote by $n^{\prime}$ and $n$ their respective length and set $\Delta=n^{\prime}-n$. We have
\begin{itemize}\setlength{\itemindent}{-1em}
\item if $\Delta=0$ and $i\geq1$, then $q(S,S^{\prime})=\frac{1}{2}\,\frac{1}{n+1}\,\frac{1}{2\pi}\,f^{\varepsilon g}_{HG}\lp\theta_{i}^{\prime}\rp$;

\item if $\Delta=0$ and $i=0$, then $q(S,S^{\prime})=\frac{1}{2}\,\frac{1}{n+1}\,\mu e^{-\mu r_{0}^{\prime}}$;

\item if $\Delta<0$, then  $q(S,S^{\prime})=\frac{1}{2}\,\zeta(n)$;

\item if $\Delta>0$, then 
$ q(S,S^{\prime})=\frac{1}{2}\,\zeta(n)  e^{-\mu \sum_{k=1}^{\Delta} r_{n+k}^{\prime}}\,\lp\frac{\mu}{2\pi}\rp^{\Delta}  \prod_{k=1}^{\Delta} f^{\varepsilon g}_{HG}\lp\theta_{n+k}^{\prime}\rp$.
\end{itemize}
From these formulas, it is straightforward to recover the acceptance probability. \\

The idea behind this mixture of mutations is to find a compromise between large jumping size of the Markov chain which implies a lot of ``burnt'' samples, and smaller jumps which provide more correlated samples, hence a worse convergence. The rotations lead to a good exploration of the domain at low cost, whereas the addition-deletion mutations ensure the visit of the whole state space $\mathcal{A}$ with $W_{0}=\omega_{0}$. The use of a perturbed phase function decreases the acceptance probability of the mutations and thus, increases the number of samples needed in order to converge to the invariant measure. But, it allows a better exploration of the domain and this why the parameter $\varepsilon$, as well as the sizes $j, J$, need to be adapted on a case by case basis. Finally, we can prove that, with this rule of mutations, Algorithm~\ref{algo:MH} produces a Markov chain that satisfies the LLN. This guarantees the convergence of the estimator defined in~\eqref{eq:def:estima:MH}.

\begin{Prop}\label{prop:convergence:MH}
If the chain $(S(t))_{t\in\N^{*}}$ is obtained by Algorithm~\ref{algo:MH} with the mutation rule given in Definition~\ref{def:mutations:rule}, then it is Harris positive with respect to the measure $\nu|_{W_{0}=\omega_{0}}$ and the estimator  $\widehat{L}_{\text{\emph{MH}}}(x_k)$ defined in~\eqref{eq:def:estima:MH} is strongly consistent for all $0\leq k\leq K-1$.
\end{Prop}
\begin{proof}
The fact that $\nu|_{W_{0}=\omega_{0}}$  is an invariant measure of $(S(t))_{t\in\N^{*}}$ is an inherent property of Metropolis-Hastings algorithm (\cite{RobertsRosenthal2006, Tierney1994}). The Harris recurrence is then obtained by checking that the chain is irreducible with respect to $\nu|_{W_{0}=\omega_{0}}$, see~\cite[Corollary 2]{Tierney1994}.  Let $\tau_{A}=\inf\lac t\in\N^{*}: S(t)\in A\rac$ denote the hitting time of any $A\subset\mathcal{A}$ such that $\nu|_{W_{0}=\omega_{0}}(A)>0$. We must demonstrate that $(S(t))_{t\in\N^{*}}$ is irreducible with respect to $\nu|_{W_{0}=\omega_{0}}$, that is,
\begin{equation}\label{eq:cdt-nu-iredducible}
\Prob_{s}\lp\tau_{A}<+\infty\rp>0, \qquad \text{ for all } s\in\mathcal{A},
\end{equation}
where $\Prob_{s}(S(1)=s)=1$. Furthermore, notice that it is sufficient to check this property for subsets $A$ of the type
\begin{equation}\label{eq:A-cylindrical}
A=\{\omega_{0}\}\times I_0 \times \prod_{i=1}^{n} \lp  I_{i}\times E_{i}  \rp,
\end{equation}
where $n\in \N^{*}$ and where for all $1\leq i\leq n$, the sets $I_{i}\subset \R_{+}$ and $E_{i}\subset \Ss^{2}$ are all sets of positive Lebesgue measure. 
 
In order to prove relation \eqref{eq:cdt-nu-iredducible} for sets of the form \eqref{eq:A-cylindrical}, consider the conditional measure $\nu_{n}|_{W_{0}=\omega_{0}}$ using~\eqref{eq:def:nun}. 
By assumptions on $A$, we see easily that $\nu|_{W_{0}=\omega_{0}}(A)>0$. 
Now, notice that by the Markov property, if $\tau_{A}$ and $\tau_{\{\omega_{0}\}\times I_{0}}$ denote respectively the time for the chain to be in $A$, resp.~in $\{\omega_{0}\}\times I_{0}$, then we have
$$\Prob_{s}\lp\tau_{A}<+\infty\rp\geq \Prob_{s}\lp\tau_{\{\omega_{0}\}\times I_{0}}<+\infty\rp  \Prob_{\{\omega_{0}\}\times I_{0}}\lp\tau_{A}<+\infty\rp, \qquad \text{ for all } s\in\mathcal{A}.$$
Now we can lower bound the right hand side of this relation in the following way:

\noindent
\emph{(i)}
We have that $\Prob_{s}\lp\tau_{\{\omega_{0}\}\times I_{0}}<+\infty\rp$ is greater than the probability of deleting all the edges of $s$ except $(r_{0},\omega_{0})$ and of modifying its length so that $r_{0}^{\prime}\in I_{0}$. This probability is strictly positive, as well as its acceptance. Indeed, we use here the fact that $j$ and $J$ are coprime (through Bezout's lemma) plus elementary relations for uniform distributions to assert that the probability of deleting all the edges is strictly positive. The positivity of acceptance is due to absolute continuity properties of $q(s,\cdot)$.

\noindent
\emph{(ii)}
The same kind of argument works in order to lower bound $\Prob_{\{\omega_{0}\}\times I_{0}}\lp\tau_{A}<+\infty\rp$.  Namely, this quantity is greater than the probability to construct directly a ray $s\in A$, which is itself strictly positive. Indeed, since $j$ and $J$ are mutually prime, it is possible to construct a ray of any desired length. Moreover, at each step, the probability of adding an edge $(r_{i}^\prime, \omega_{i}^\prime)\in I_{i}\times E_{i}$, as well as its acceptance are always strictly positive. 

\noindent
We have thus obtained that $ \Prob_{s}\lp\tau_{\{\omega_{0}\}\times I_{0}}<+\infty\rp  \Prob_{\{\omega_{0}\}\times I_{0}}\lp\tau_{A}<+\infty\rp>0$,
which concludes the proof.
\end{proof}

\begin{Rem}
The process $\lp N_{t}\rp_{ t\geq1}$  that gives the length of the ray $S(t)$ at time $t$ behaves like a birth-death process with time inhomogeneous rates. If there exists an invariant measure of the process $\lp N_{t}\rp_{ t\geq1}$, then it must coincide with the negative binomial distribution ${\rm NB}(1,\rho)$ that drives the length of a path $S\sim\nu$. This provides an easy criterium in order to check that the chain has already mixed, for example with a chi-squared test on the empirical distribution of $\lp N_{t}\rp_{ t\geq1}$.
\end{Rem}


\section{Simulation and comparison of the methods}\label{sec:comparison}
In this section, we compare the estimates of the fluence rate $L(x_{k})$ provided by three methods: Monte Carlo with Wang-Prahl algorithm (denoted by WANG, see~\cite{Prahl1988, Wang1992_rep}),   MC-SOME (see~\eqref{eq:def:estima:MCSome}) and  the Metropolis-Hastings (MH) (see~\eqref{eq:def:estima:MH}) with the mutation rule given in Definition~\ref{def:mutations:rule}. We tested the methods in different settings. Here, we present results in a framework corresponding to a healthy homogeneous rat brain tissue. We chose to follow~\cite{AngellPetersen2007} for the values of the optical parameters (see Table \ref{table:parameters:value}). Other values for rat or human brain can be found in~\cite{Bevilacqua1999, CheongPrahl1990, Johns2005}. The volume of the cube $V$ equals $8 \text{ cm}^{3}$, that is $V=[-1,1]^{3}$. It is discretized into voxels whose volume is $\lp 0.04\rp^{3}\,\text{cm}^{3}$. The half-opening angle of the optical fiber was set to $\alpha=\frac{\pi}{10}$ and the constant $c$ in~\eqref{eq:luminance:proba:repres} was set $c=1$.

We chose the following simulation parameters for the three methods so that they need the same amount of computational time. Those are
\begin{description}
\item[WANG:] $M=6000$ photons trajectory.
\item[MC-SOME:] $M=30000$ rays, $M_{\text{points}}=40$ points chosen in each ray and $M_{\text{rot}}=30$ rotations with respect to the initial direction. 
\item[MH] $j=10$, $J=21$, $\varepsilon=0.9$, $T=250000$ steps of the chain and  $M_{\text{rot}}=30$ rotations with respect to the initial direction. \\
\end{description}

In Fig.~\ref{fig:comparison:contourplot}, we picture for each methods, a zoom of the contour plot of the estimates in the plane $x=-0.04$.  
%
%
%
%
%
%
In Fig.~\ref{fig:comparison:estimatesalongY}, we compare the estimates along several lines of voxels $(\ell_{i})_{i=1,\ldots,6}$ which are parallel to the $y$-axis and pass through the points $(0,0,-0.08)$, $(0,0,-0.12)$, $(0,0,-0.16)$, $(0,0,-0.4)$, $(0,0,-0.48)$ and $(0,0,-0.6)$ respectively (see Fig. \ref{fig:comparison:choicevoxels}).  
We notice that  MC-SOME gives more consistant estimates than the two other algorithms whose estimates are more noisy. Moreover, it seems that the MH-estimates are not very symmetric. Perhaps because the algorithm had not converged yet.

\begin{figure}[h]
\centering
   \includegraphics[scale=0.2]{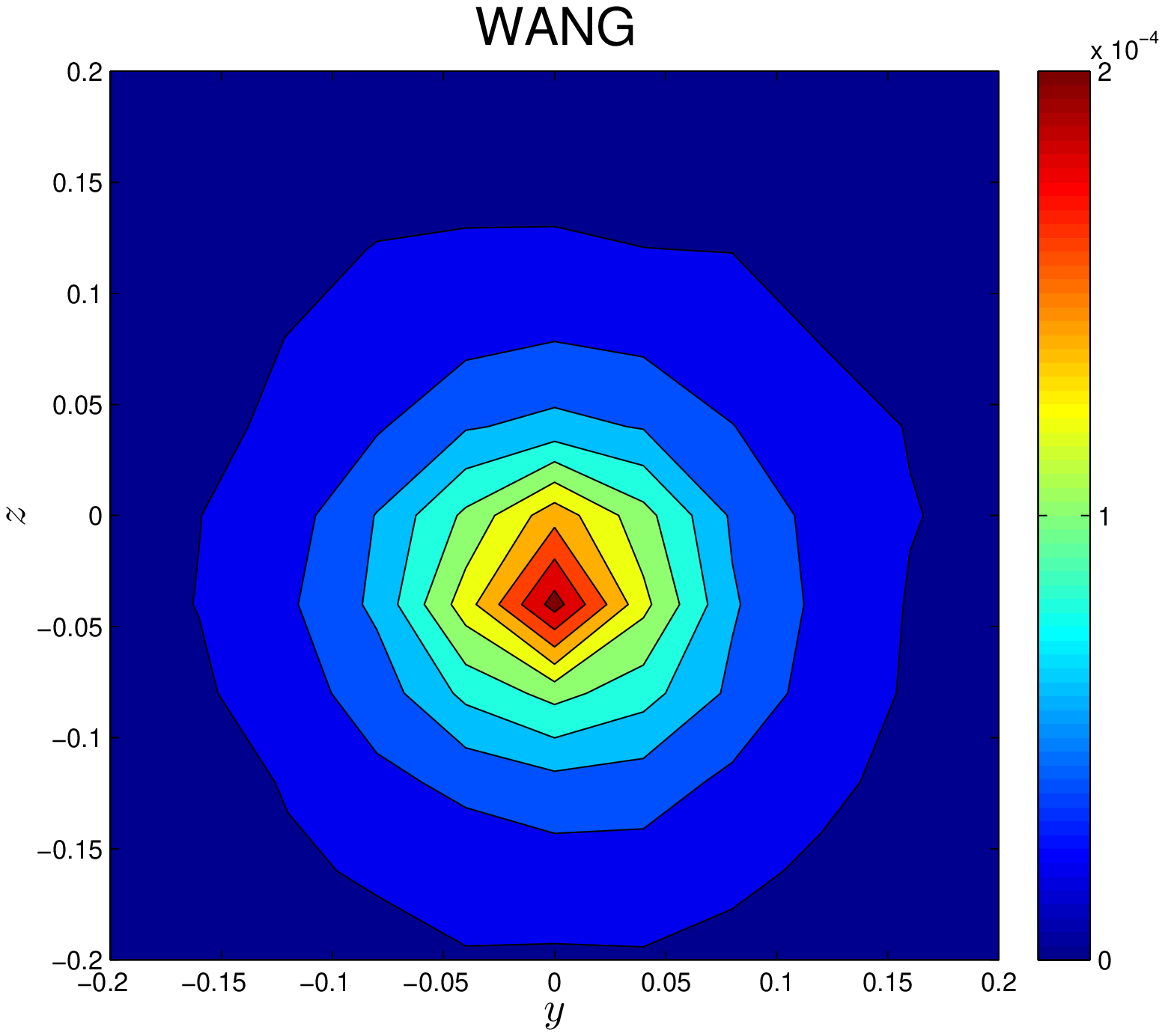}
 \includegraphics[scale=0.2]{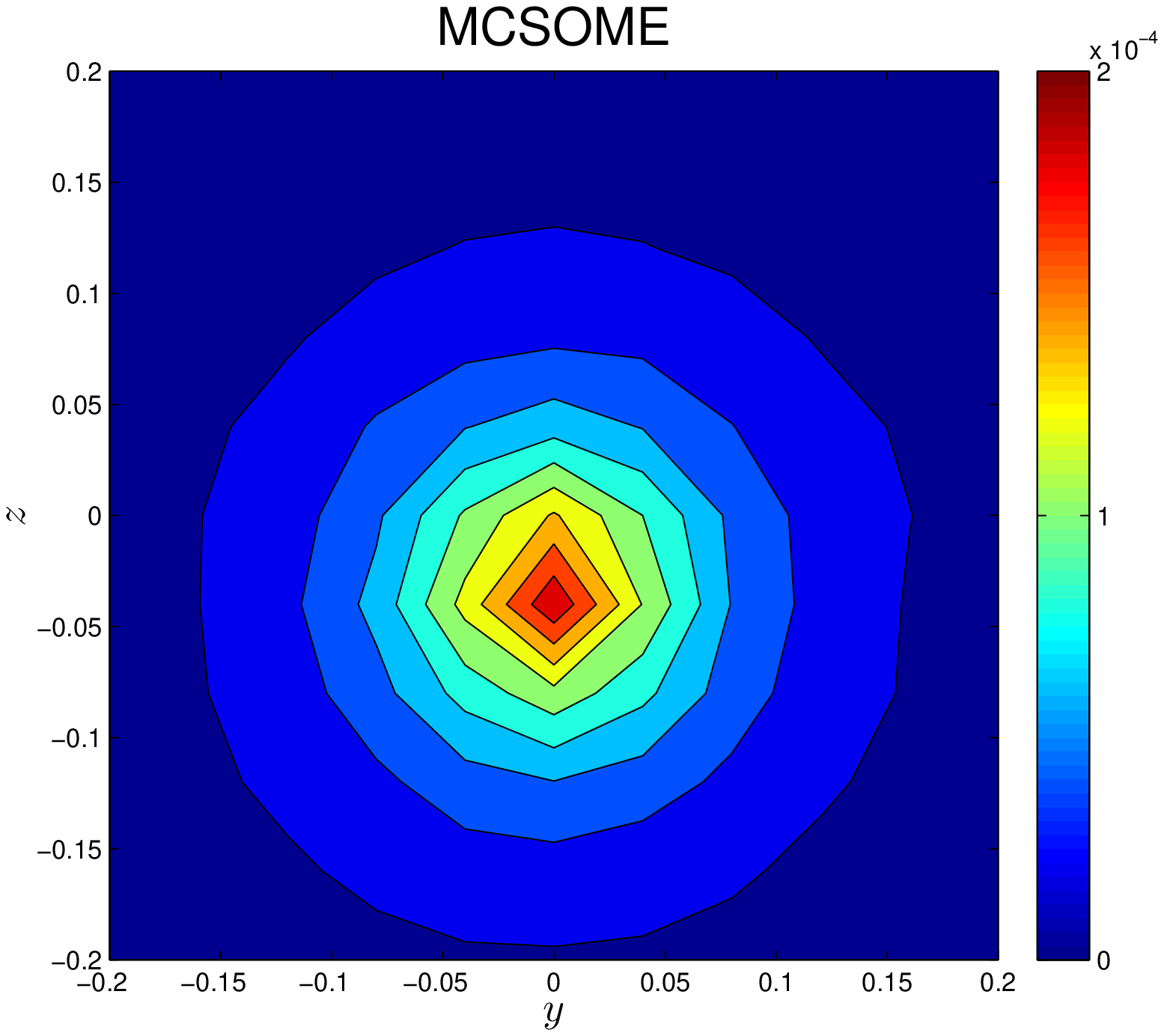}
 \includegraphics[scale=0.2]{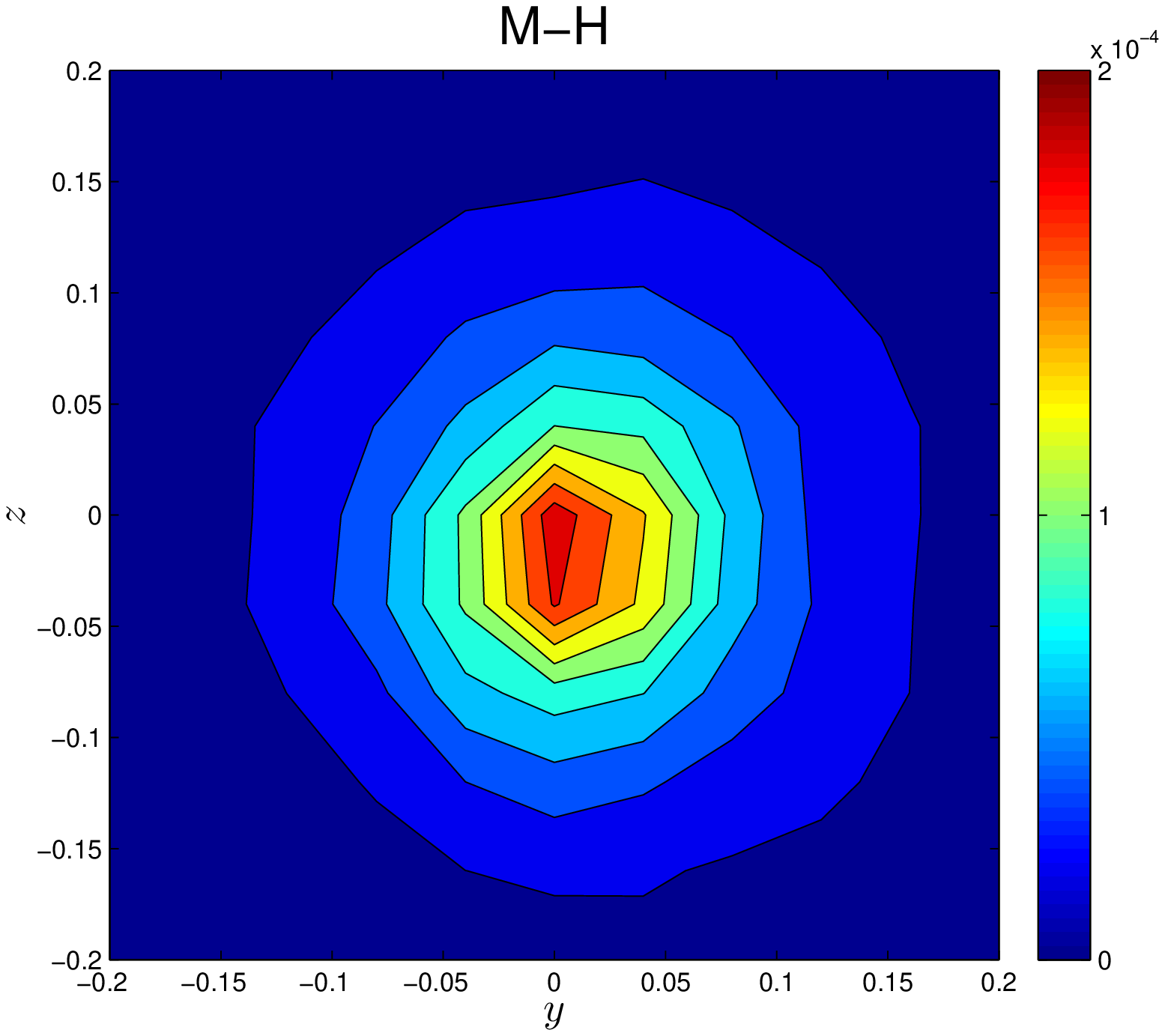}
\caption{Contour plots of the fluence rate estimates in the plane $x=-0.04$ for WANG,  MC-SOME and MH.}
\label{fig:comparison:contourplot}
\end{figure}

\begin{figure}[h!]
\begin{center}
\includegraphics[width=.45\textwidth]{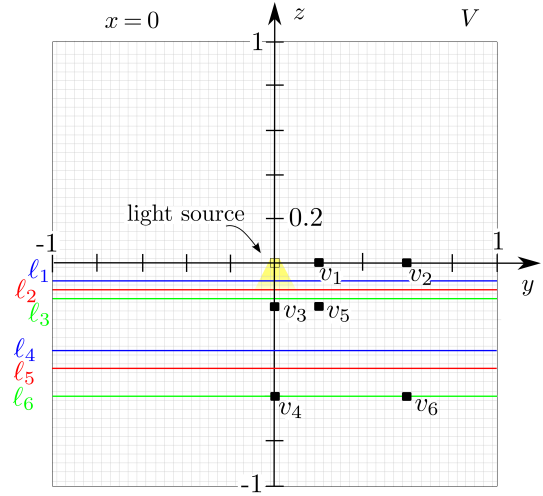} 
\caption{Choice of six particular voxels and position of the lines $(\ell_{i})_{i=1,\ldots, 6}$ in the cube $V$.}
\label{fig:comparison:choicevoxels}
\end{center}
\end{figure}

\vspace{0.2cm}

\begin{figure}[h!]
\begin{center}
\includegraphics[width=\textwidth]{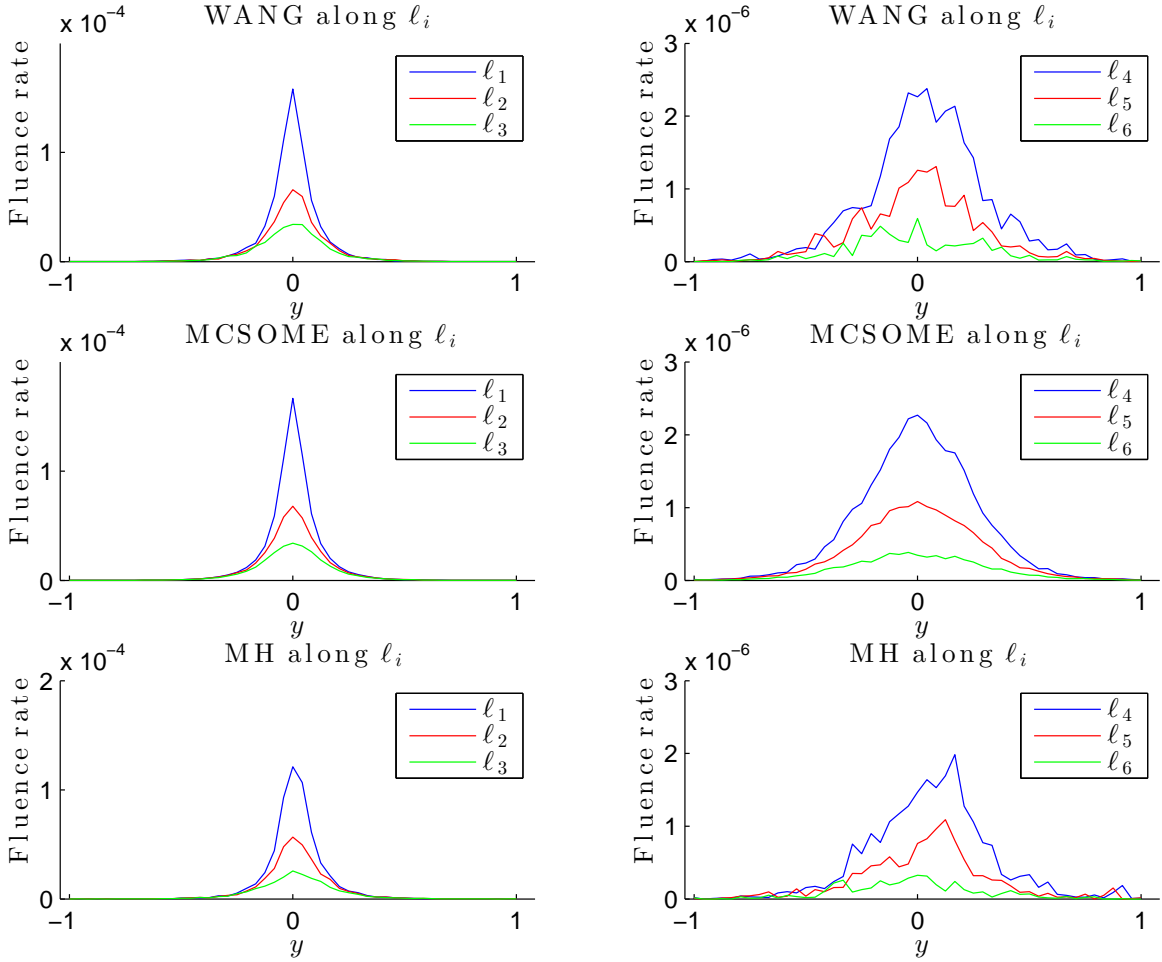} 
\caption{Fluence rate estimates along the lines $(\ell_{i})_{i=1,\ldots, 6}$ with WANG,  MC-SOME and MH.}
\label{fig:comparison:estimatesalongY}
\end{center}
\end{figure}

\newpage
{Let us conclude this section by studying the accuracy of the methods 
 by mean of $50$ independent replicates of these estimates.} In Fig.~\ref{figure:boxplot:comparison:WangMCSOME}, the boxplots compare the dispersion of the $50$ estimates of each method in six voxels $(v_{i})_{i=1,\ldots,6}$ such that (see Fig.~\ref{fig:comparison:choicevoxels})
\begin{align}\label{eq:voxelchoice}
&(0,0.2,0)\in v_{1},&& (0,0.6,0)\in v_{2},&&(0,0,-0.2)\in v_{3},\nonumber\\
&(0,0,-0.6)\in v_{4},&&(0,0.2,-0.2)\in v_{5},&& (0,0.6,-0.6)\in v_{6}.
\end{align}
 \noindent
On one hand, we see that  MC-SOME is much more consistent than WANG, because of the variance reduction that we use in MC-SOME. On the other hand, MH gives spread estimates. This is due to the simulations for which the Markov chain has not yet converged. In Table~\ref{table:comparison:WangMCSOME}, we have the mean of the estimates and their mean square error in each of the $6$ voxels.
\begin{figure}[h!]
\begin{center}
\includegraphics[width=0.45\textwidth]{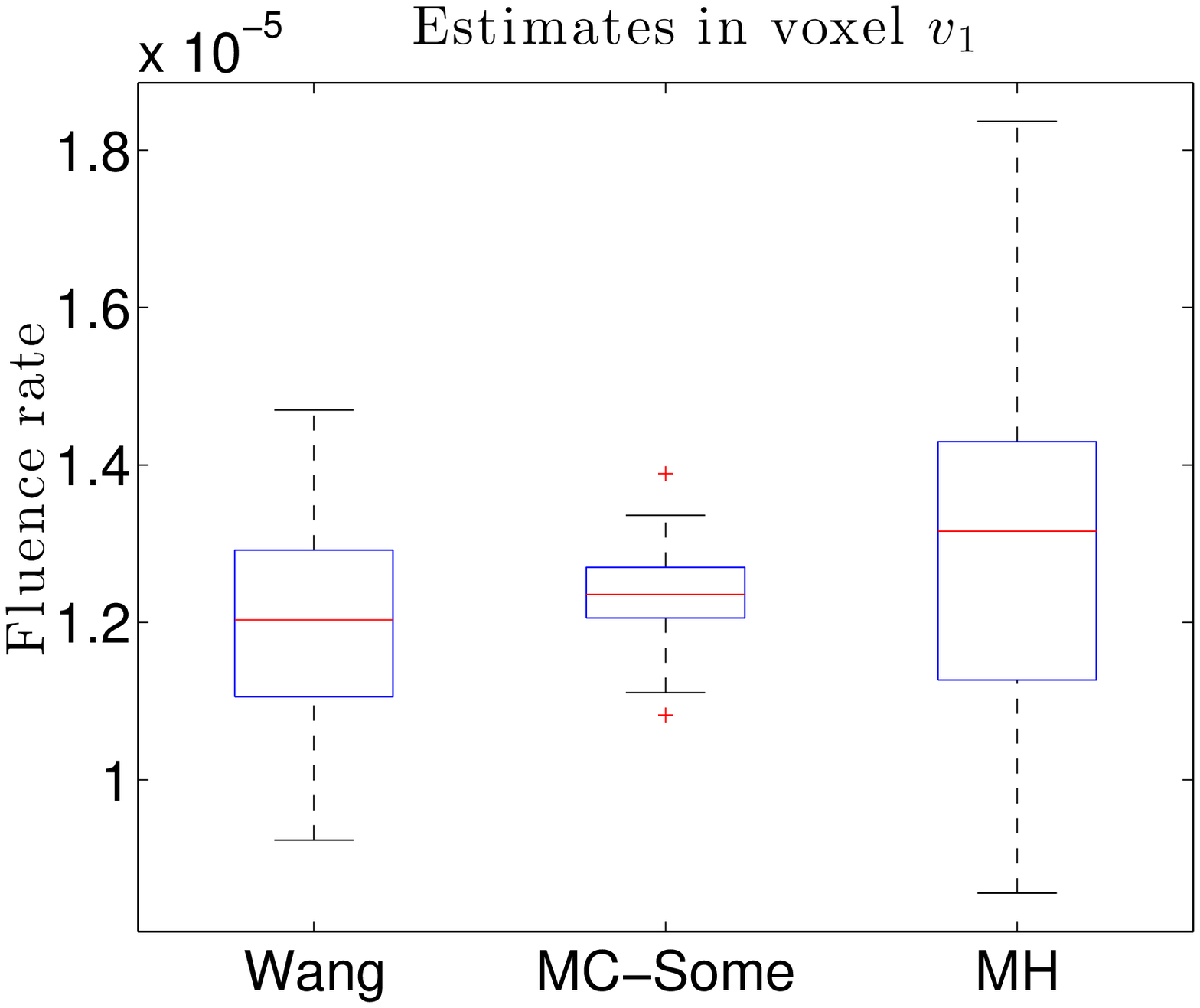}\hspace{3ex} \includegraphics[width=0.45\textwidth]{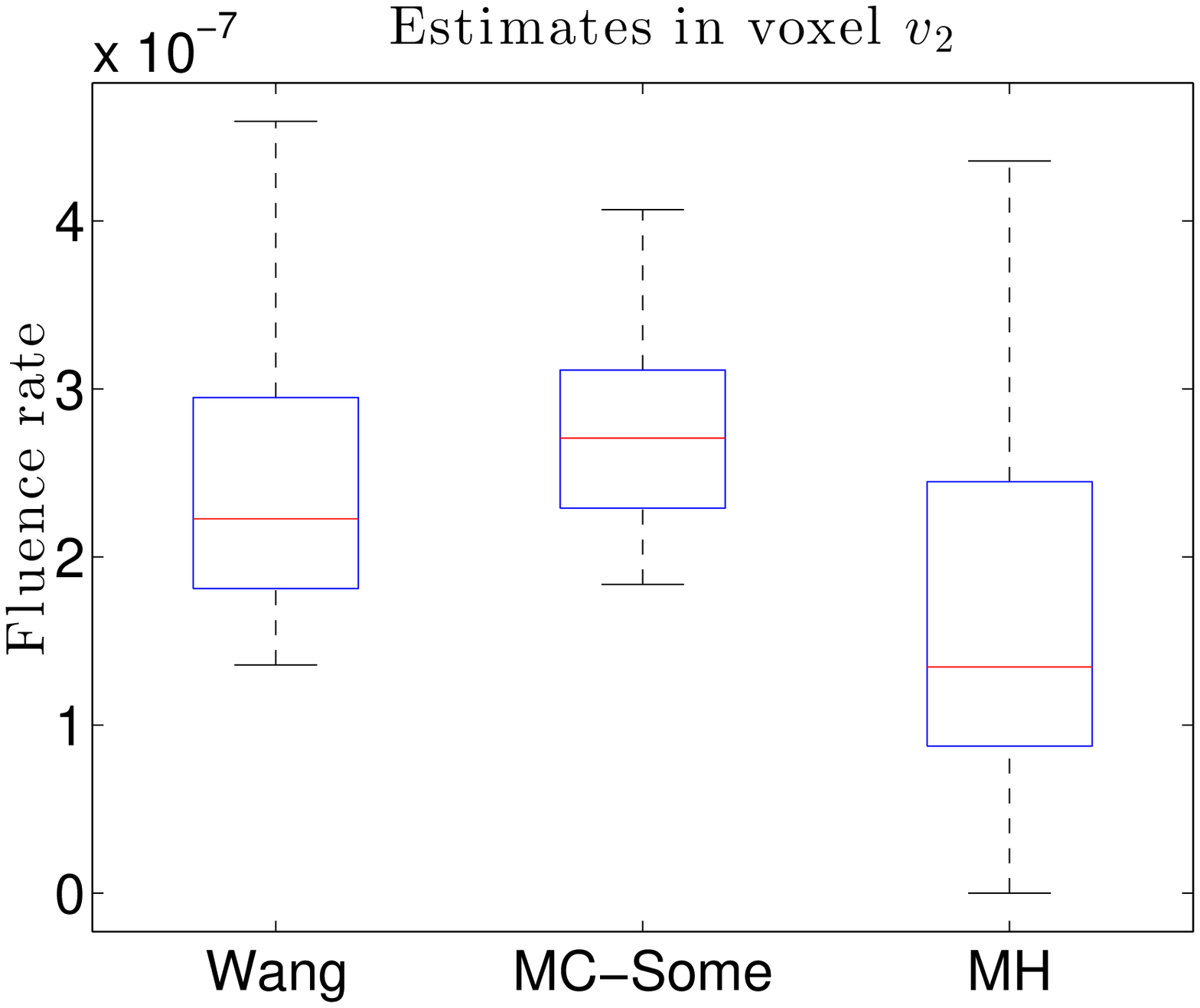}\\[6ex]
\includegraphics[width=0.45\textwidth]{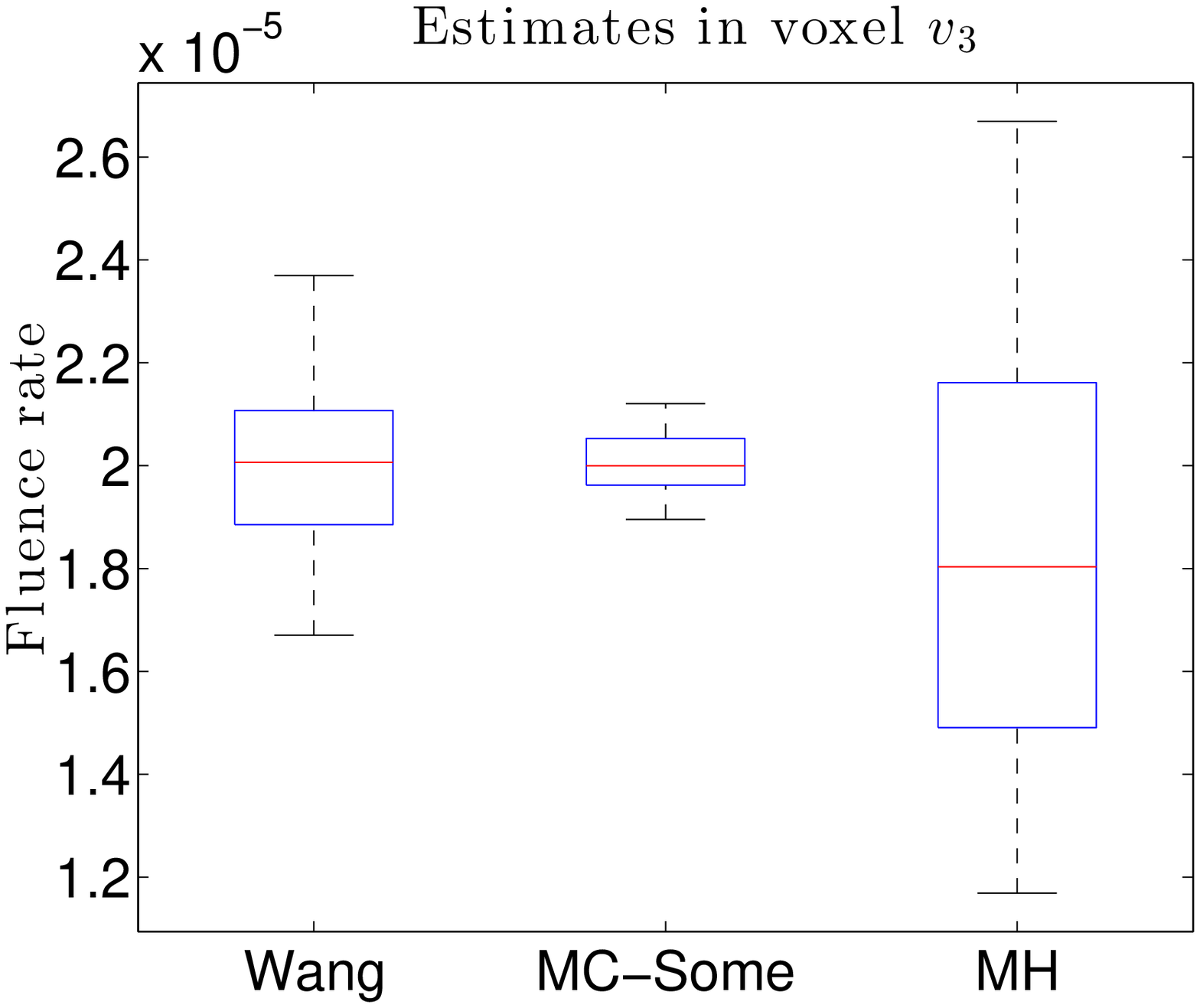}\hspace{3ex} \includegraphics[width=0.45\textwidth]{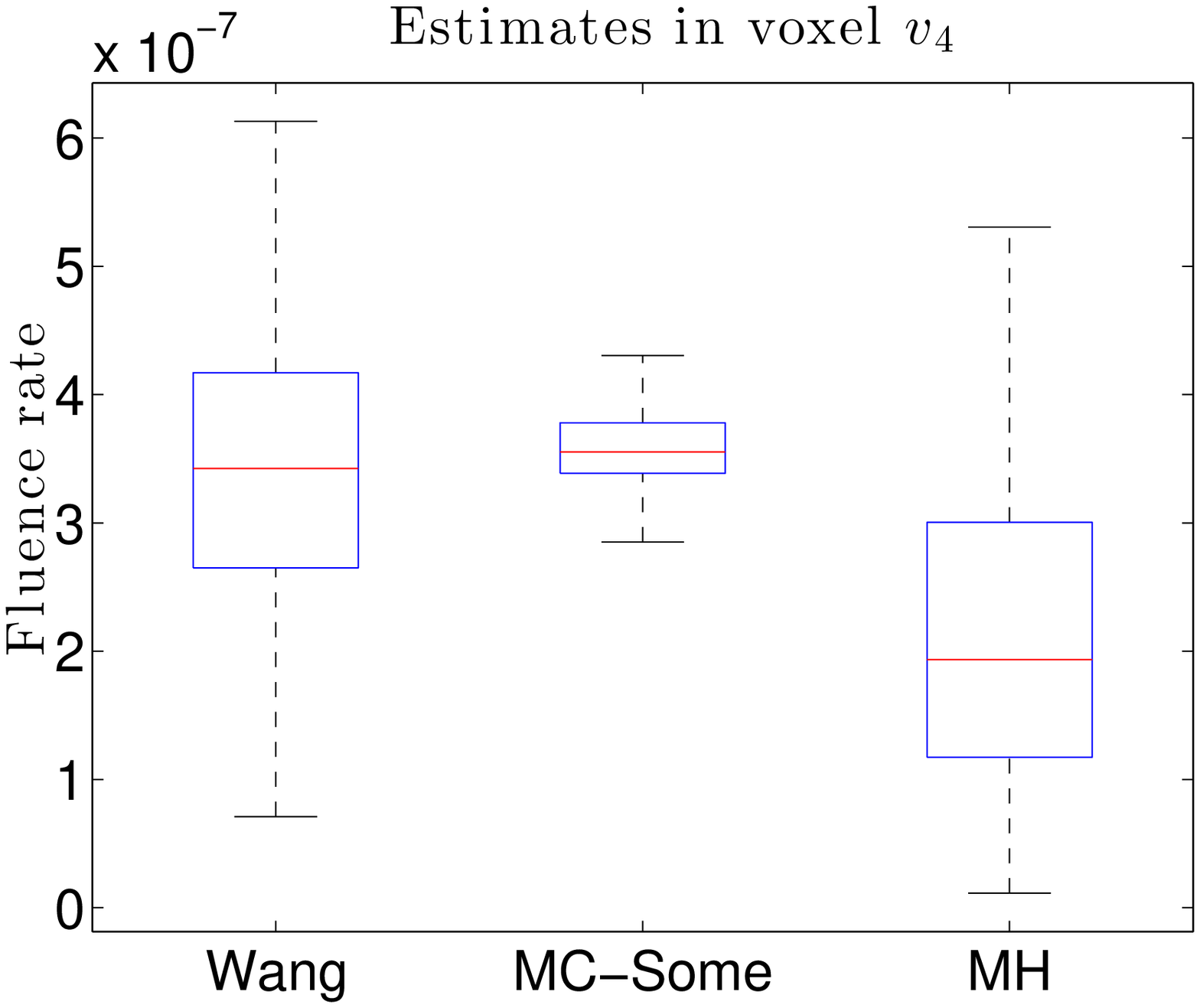}\\[6ex]
\includegraphics[width=0.45\textwidth]{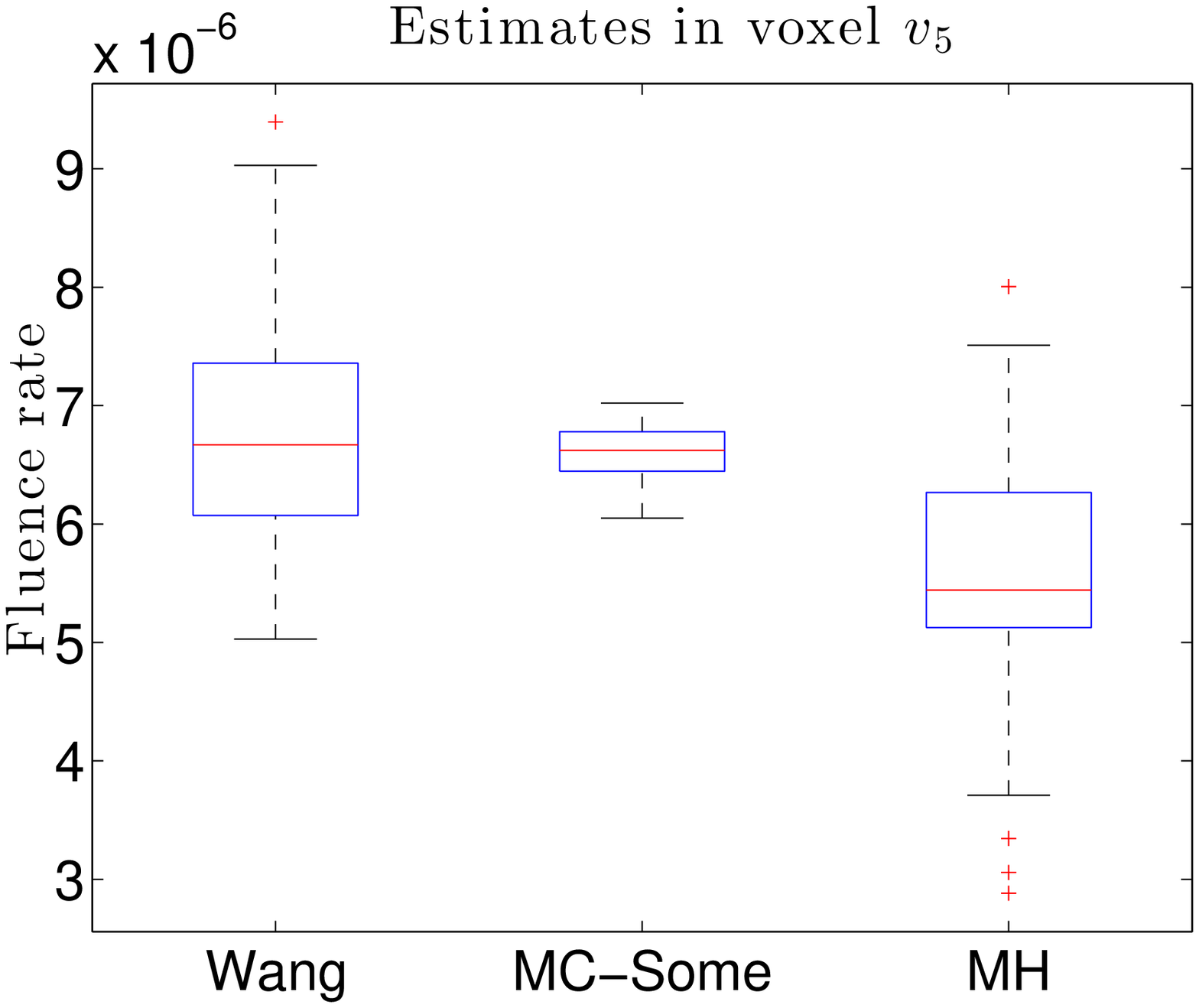}\hspace{3ex} \includegraphics[width=0.45\textwidth]{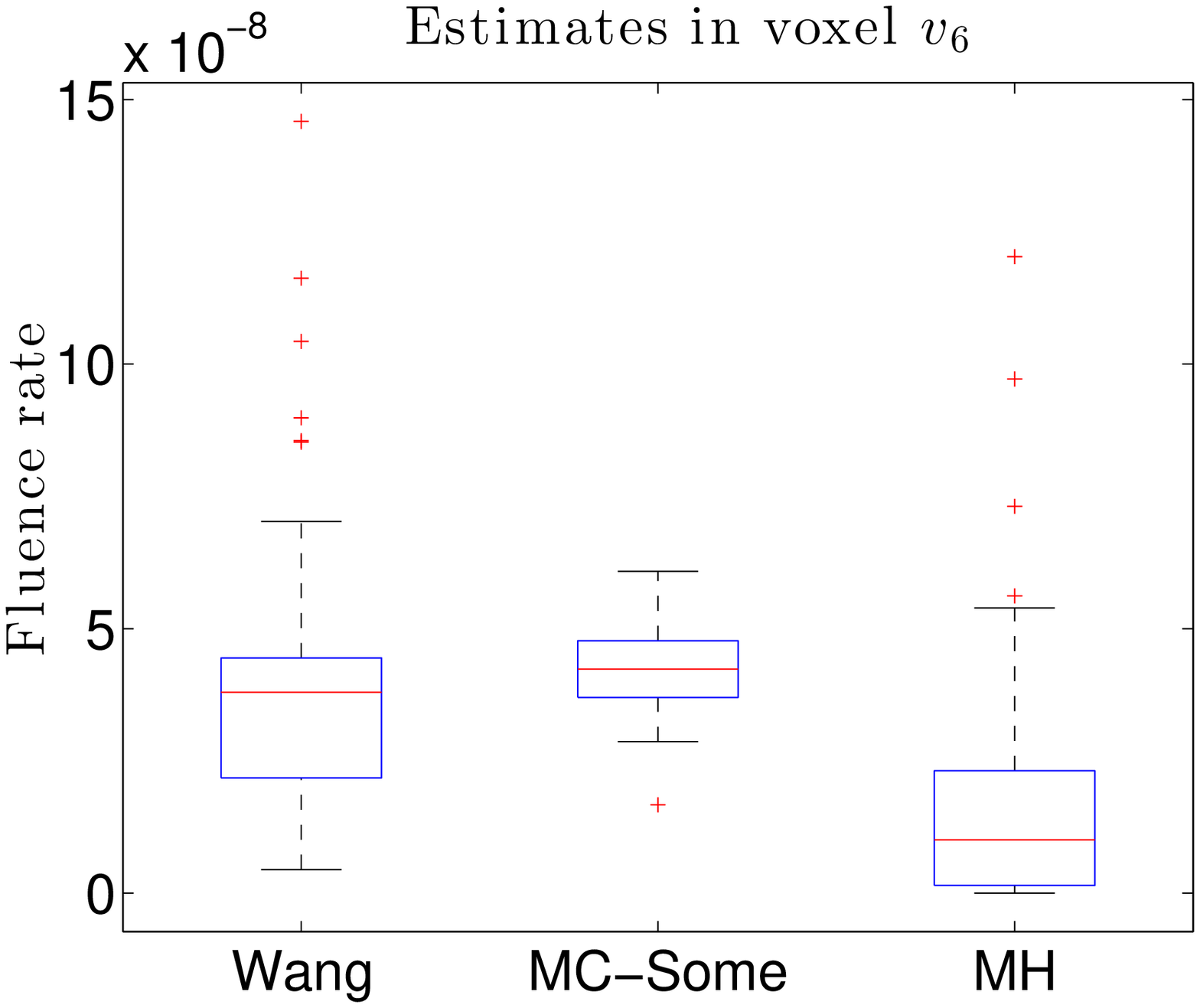}\\[4ex]
\caption{Boxplots of $50$ fluence rate estimates in the voxels $(v_{i})_{i=1,\ldots,6}$ with WANG, MC-SOME and MH.}\label{figure:boxplot:comparison:WangMCSOME}
\end{center}
\end{figure}
\begin{table}[h!]
\centering
\begin{center}
\begin{tabular}{|c|c|c|c|c|c|c|}
\hline
& \multicolumn{3}{c|}{Mean} & \multicolumn{3}{c|}{Mean Square Error} \\
  \cline{2-7}
	& Wang  &  MC-SOME & M-H  			& Wang 	& MC-SOME   	& M-H\\
   \hline
$v_{1}$ & $1.203\cdot 10^{-5}$ & $1.2366\cdot 10^{-5}$ & $1.3083\cdot 10^{-5}$ & $1.4484\cdot 10^{-12}$ & $3.9108\cdot 10^{-13}$ & $ 4.5158\cdot 10^{-12}$\\
$v_{2}$ & $2.4309\cdot 10^{-7}$ & $2.7177\cdot 10^{-7}$ & $1.6545\cdot 10^{-7}$ & $6.3788\cdot 10^{-15}$ & $ 2.5597\cdot 10^{-15}$ & $1.4411\cdot 10^{-14}$\\
$v_{3}$ & $2.0014\cdot 10^{-5}$ & $2.0033\cdot 10^{-5} $ & $1.8337\cdot 10^{-5}$ & $2.503\cdot 10^{-12}$ & $3.0408\cdot 10^{-13}$ & $1.6200\cdot 10^{-11}$\\
$v_{4}$ & $ 3.4947\cdot 10^{-7}$ & $3.5713\cdot 10^{-7}$ & $2.1497\cdot 10^{-7}$ & $1.4136\cdot 10^{-14}$ & $8.1737\cdot 10^{-16}$ & $1.5819\cdot 10^{-14}$\\
$v_{5}$ & $6.7519\cdot 10^{-6}$ & $6.6047\cdot 10^{-6}$ & $5.4786\cdot 10^{-6}$ & $8.0457\cdot 10^{-13}$ & $4.4781\cdot 10^{-14}$ & $1.1378\cdot 10^{-12}$\\
$v_{6}$ & $3.9952\cdot 10^{-8}$ & $4.217\cdot 10^{-8}$ & $1.8755\cdot 10^{-8}$ & $8.4058\cdot 10^{-16}$ & $6.6977\cdot 10^{-17}$ & $6.5413\cdot 10^{-16}$\\
\hline
\end{tabular}
\end{center}
\caption{Mean and mean square error of the $50$ fluence rate estimates in 6 voxels with WANG, MC-SOME and MH.}
\label{table:comparison:WangMCSOME}
\end{table}

\section{Inverse problem and sensitivity}\label{sec:inverse}

For biologists, it is of considerable practical importance to have good estimates of the optical coefficients of the tissue they consider. One way to do this estimate is to compare simulated data with measurements of the fluence rate in the tissue and adjust the optical parameters of the simulation until obtaining values close to the measurements. Thanks to the probabilistic representation \eqref{eq:luminance:proba:repres}, this problem can be numerically solved as we shall see. 

\subsection{Sensitivity of the measurements}\label{sec:sensitivity}

As a preliminary step towards a good resolution of the inverse problem, we first observe how fluence rate measurements vary with respect to the optical parameters $g$, $\mu_{s}$ and $\mu_{a}$. To this aim, we built a small database of simulations for different values of the parameters and then compared the estimated fluence rate. The estimates are computed by resorting to MC-SOME, which is the best performing method among the three we have implemented according to Section \ref{sec:comparison}. 

First, we choose a reference simulation obtained with \emph{reference parameters} $(g^{*}, \mu_{a}^{*},\mu_{s}^{*})$. Then, we pick $n\geq 1$ voxels and consider their respective fluence rate estimates as \emph{measurements}. That is, we choose $n$ voxel centers $(x_{k_{i}})_{i=1,\ldots,n}$ and define 
$$m_{i}=\widehat{L}(x_{k_{i}};g^{*},\mu_{a}^{*},\mu_{s}^{*}), \qquad i=1, \ldots,n,$$ 
where we recall that $L(x_{k_i})$ is defined by \eqref{eq:def-L-xk} with $\gamma_{W_0}=\ind_{\{\omega_0'\in C^{2\alpha}\}}\sigma(d\omega_0')/\sigma(C^{2\alpha})$ and where we stress the dependence on the optical coefficients by writing $\widehat{L}(x_{k_{i}};g^{*},\mu_{a}^{*},\mu_{s}^{*})\equiv \widehat{L}(x_{k_{i}})$.
Now for each possible triplet of parameters $(g, \mu_{a},\mu_{s})$, we compute the normalized quadratic error (or \emph{evaluation error})
\begin{equation}\label{eq:def:J}
J(\mu_{a},\mu_{s}, g )=\frac{1}{2}\sum_{i=1}^{n} \lp\frac{\widehat{L}(x_{k_{i}};g,\mu_{a},\mu_{s})-m_{i}}{m_{i}}\rp^{2}.
\end{equation}
For the dataset of simulations, we use the same settings as in Section~\ref{sec:comparison} ($|V|=8\text{ cm}^{3}$, the volume of a voxel is $\lp0.04\rp^{3}\,\text{cm}^{3}$ and  $\alpha=\frac{\pi}{10}$). The variable parameters are:  $g$,  $\mu_{a}$ and  $ \mu_{s}$. Their values are given in Table~\ref{table:sensi:param}. This choice is motivated by~\cite{AngellPetersen2007, Bevilacqua1999, CheongPrahl1990, Johns2005}. The anisotropy parameter $g$ does not vary a lot between tissue type (healthy or tumorous) and it is often even hidden in a reduction of the scattering coefficient $\mu_{s}^{\prime}=\mu_{s}(1-g)$. For this reason, we chose only three values in a small range of common values. Concerning the other parameters, we chose five values in intervals covering values corresponding to healthy and tumorous brain tissues according to~\cite{AngellPetersen2007, Johns2005}.
\begin{table}[h!]
\begin{center}
  \begin{tabular}{|l|l|l|l|l|l|}
    \hline
     $g$ 					& $0.85$ & $0.90$ & $0.95$ & \multicolumn{2}{c|}{} \\ 
    \hline
     $\mu_{a}$ in $cm^{-1}$    	& $0.5$ & $0.75$ & $1$ & $1.25$ & $1.5$\\ 
    \hline
     $\mu_{s}$ in $cm^{-1}$ 	& $75$ & $90$ & $105$ & $120$ & $135$\\ 
    \hline
  \end{tabular}\caption{Values of the optical parameters for the study of sensitivity.}\label{table:sensi:param}
\end{center}
\end{table}

Figures~\ref{figure:colormap:sensitivity} and~\ref{figure:sensitivity:3vox:mua} give different representation of the variation of the error $J(\mu_{a},\mu_{s}, g)$ with respect to the optical parameters. The real values are $(\mu_{a}^{*},\mu_{s}^{*},g^{*})=(0.75, 105, 0.9)$ and we set $n=3$, $x_{k_{1}}\in v_{2}$, $x_{k_{2}}\in v_{4}$, $x_{k_{3}}\in v_{6}$ respectively (see Fig.~\ref{fig:comparison:choicevoxels}). We see that the sensitivity in the parameters $\mu_{s}$ and $g$ is very low compared to the sensitivity in $\mu_{a}$.  In Fig.~\ref{figure:colormap:sensitivity}, we see that a wrong value of $\mu_{a}$ has strong effects on the error function and that it becomes then almost impossible to see any tendency for the anisotropy parameter $g$. Notice also that an undervaluation of $\mu_{a}$ is worse than an overvaluation in terms of the error.

\begin{figure}[h!]
\begin{center}
\includegraphics[width=0.31\textwidth]{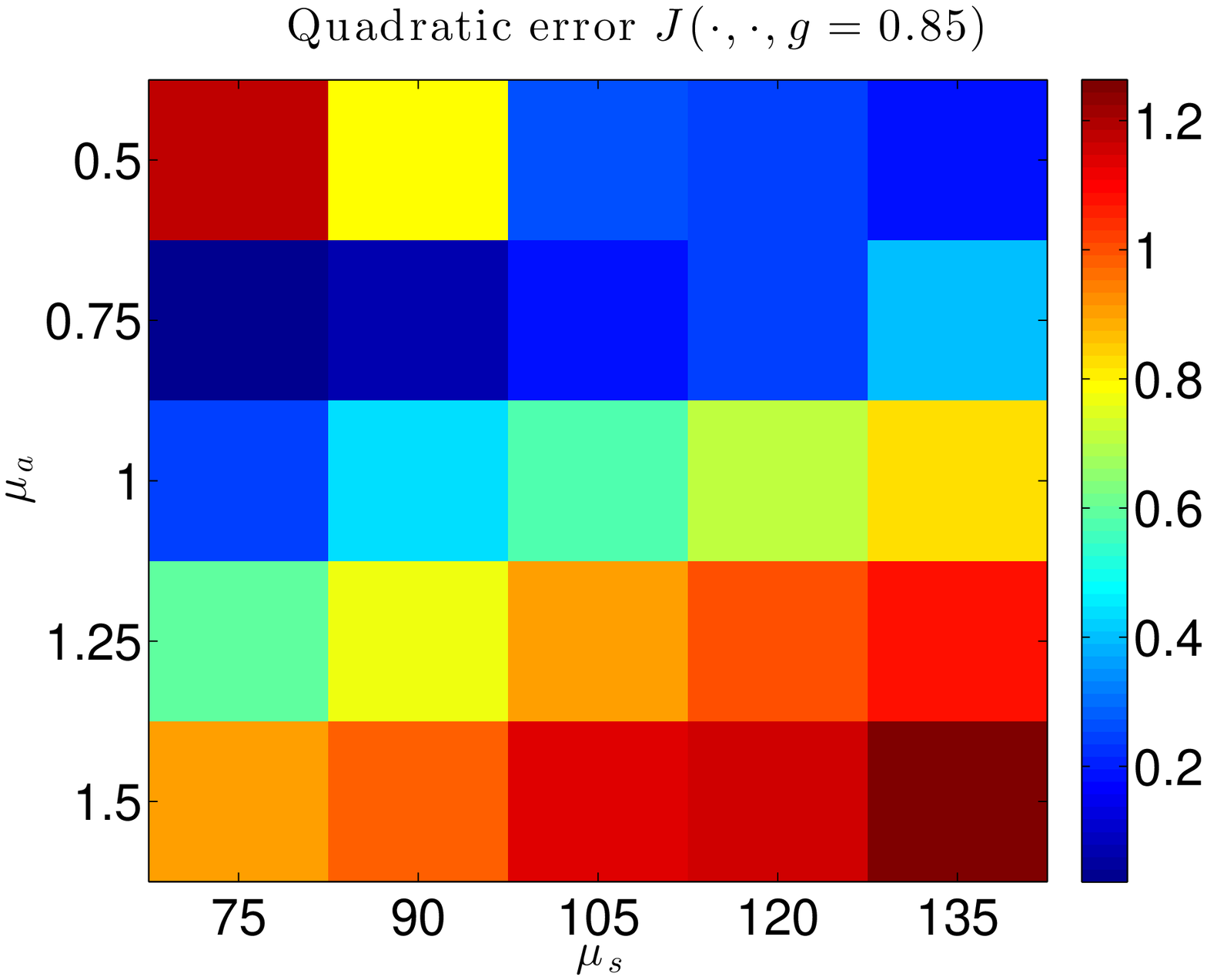}\hspace{1ex} 
\includegraphics[width=0.31\textwidth]{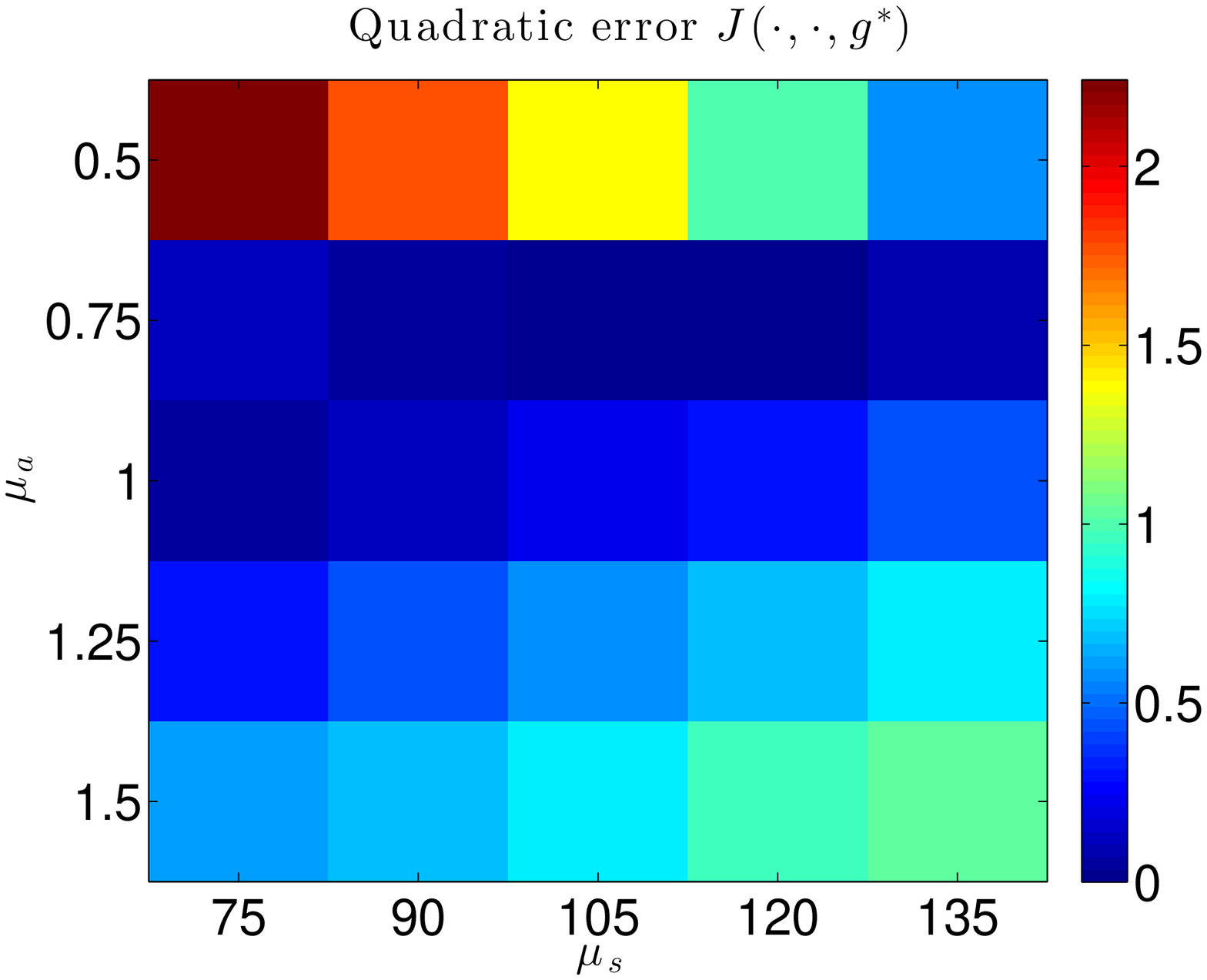}\hspace{1ex}
\includegraphics[width=0.31\textwidth]{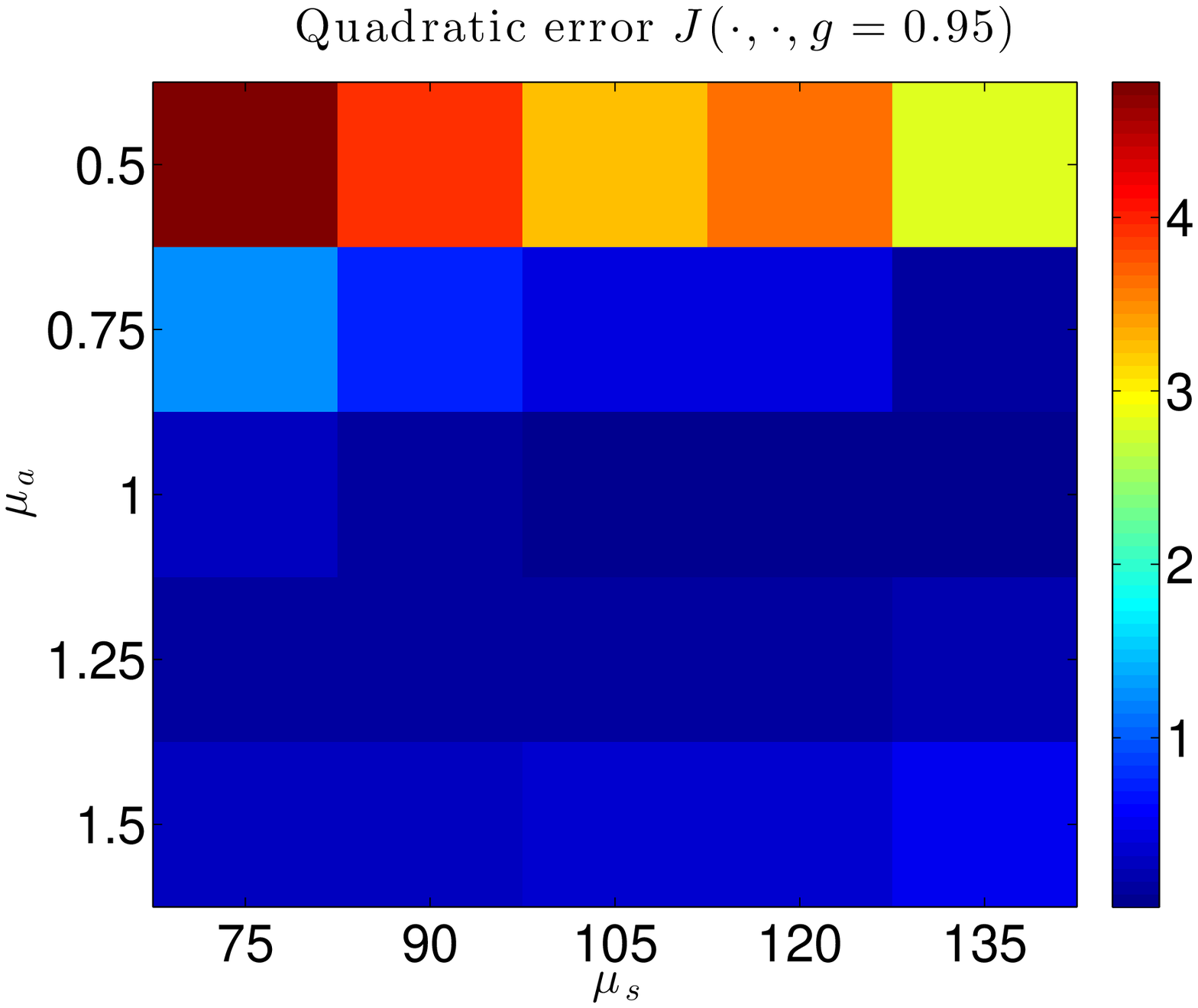}
\caption{Colormap of the quadratic error  $(\mu_{a},\mu_{s})\mapsto J(\mu_{a},\mu_{s}, g)$ for three values of $g${, where $\mu_{s}$ is displayed on the $x$ axis and $\mu_{a}$ on the $y$ axis}.}
\label{figure:colormap:sensitivity}
\end{center}
\end{figure}

\begin{figure}[h!]
\begin{center}
\includegraphics[width=0.31\textwidth]{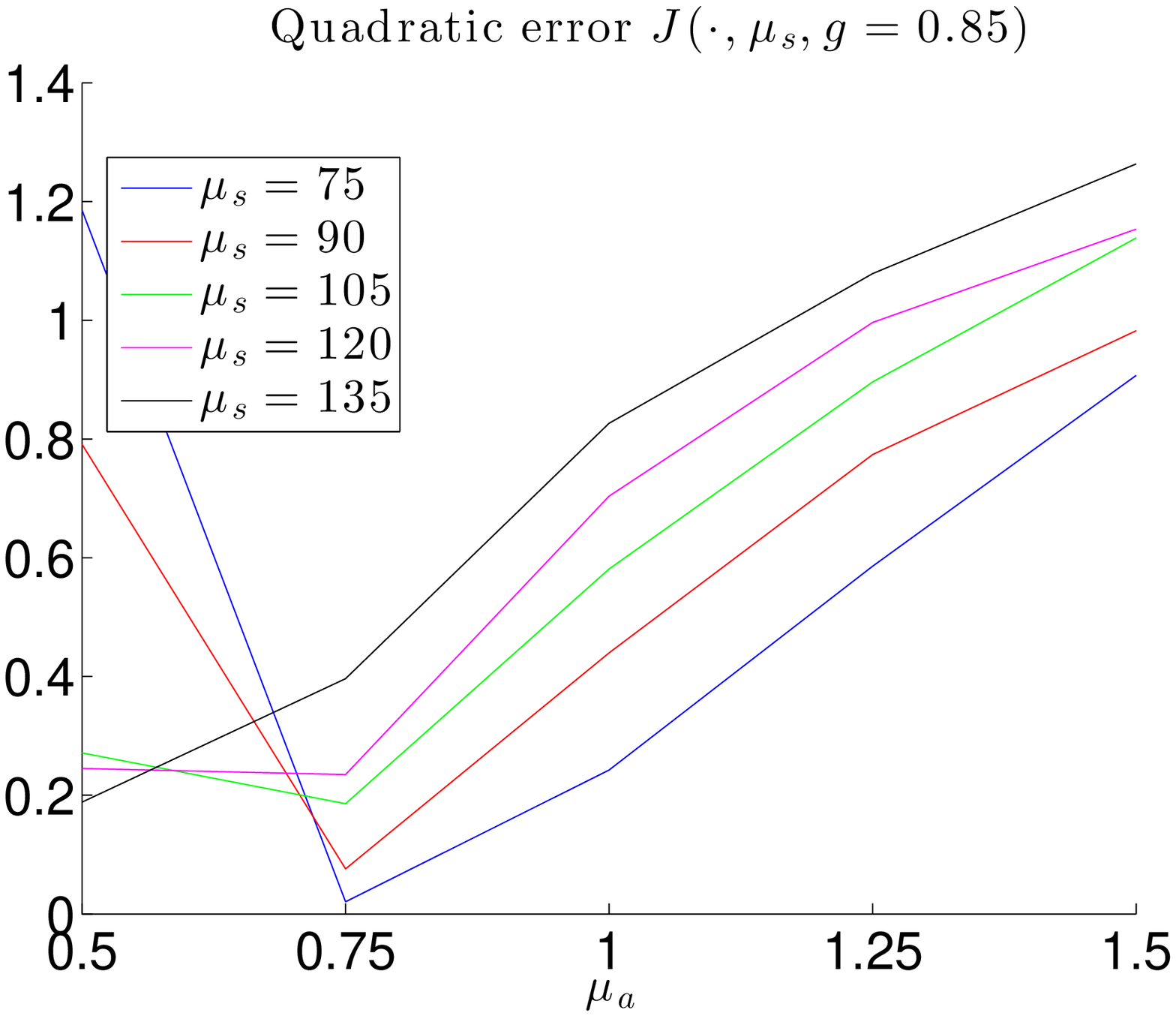}\hspace{1ex} 
\includegraphics[width=0.31\textwidth]{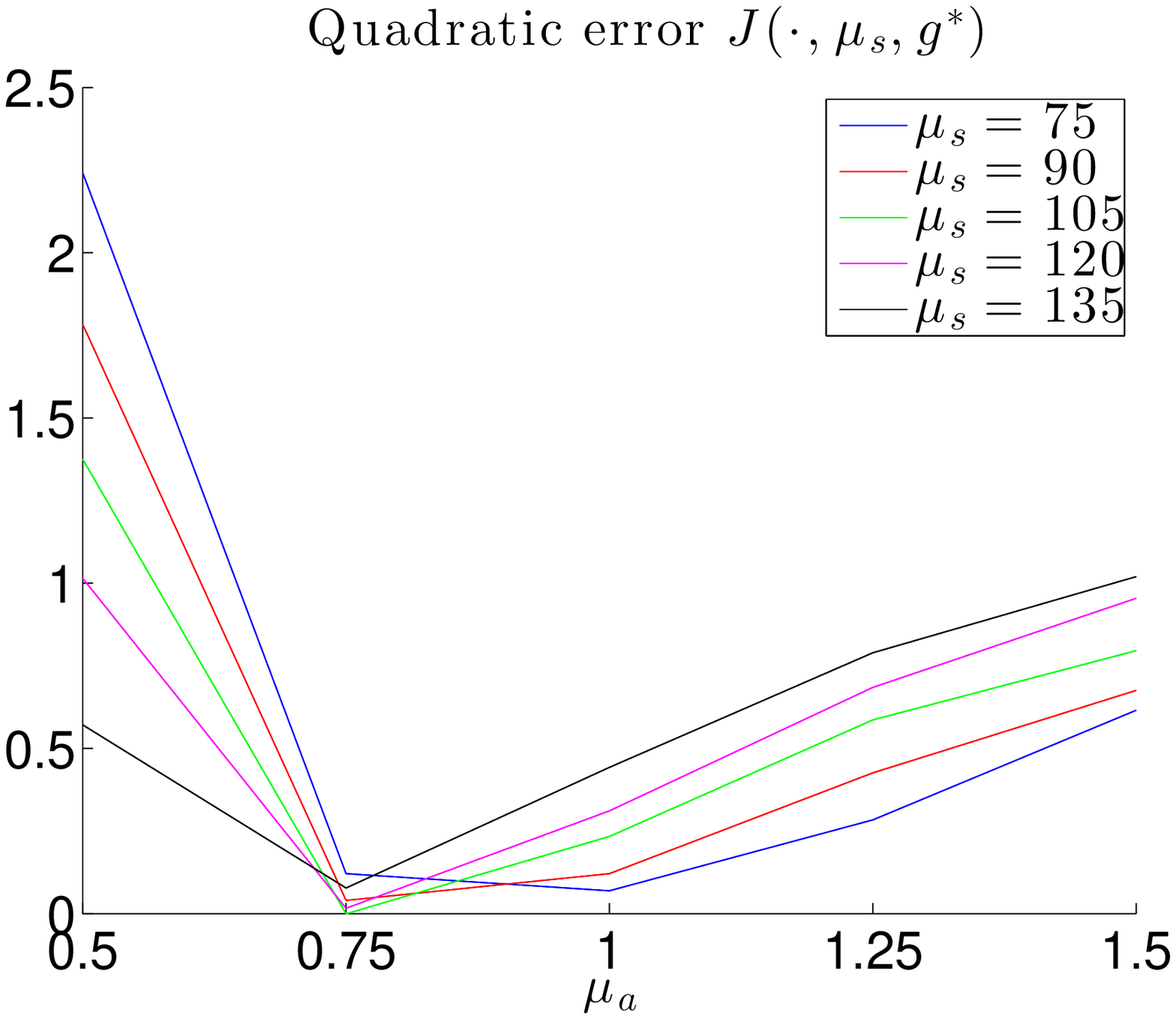}\hspace{1ex}
\includegraphics[width=0.31\textwidth]{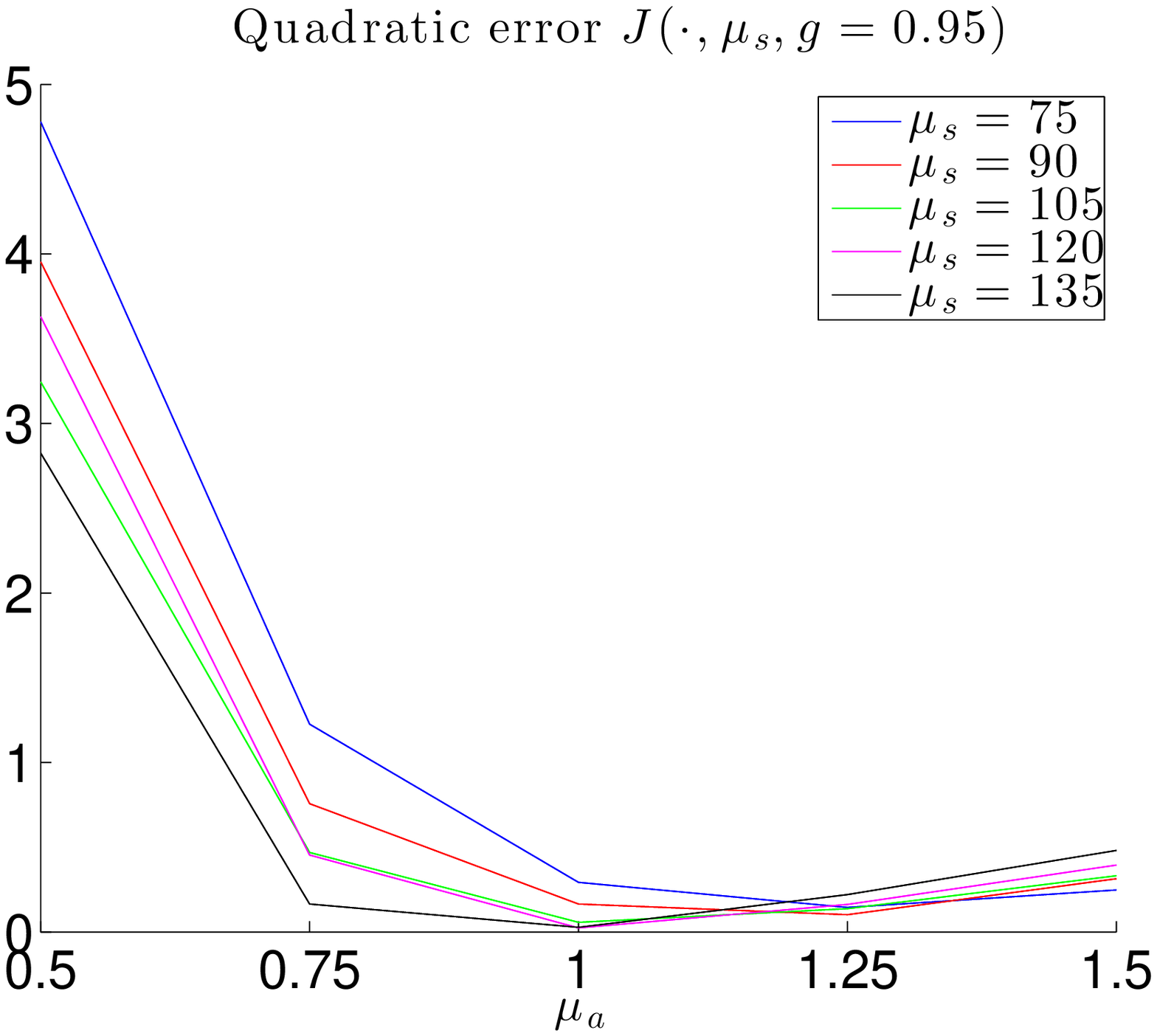}
\caption{Quadratic error $\mu_{a}\mapsto J(\mu_{a},\mu_{s}, g)$  for three values of $g$.}
\label{figure:sensitivity:3vox:mua}
\end{center}
\end{figure}


\subsection{Parameters estimation}

This section is devoted to the estimation of the parameters $\mu_{a}$ and $\mu_{s}$ only. Indeed, we  have seen in the last section that the sensitivity of the fluence rate with respect to the anisotropic parameter $g$ is low. Moreover, simulations do not show any monotonicity or tendency in the error for this parameter because, in our settings, the Monte Carlo error prevails over the evaluation error. In addition, for our purpose, the uncertainty about $g$ is small in front of the uncertainty of the two other parameters (see~\cite{AngellPetersen2007}). We shall thus suppose in the sequel that $g$ is known.

With these preliminary considerations in mind, our goal is to solve the following nonlinear least square minimization problem:
Find $(\mu_{s}, \mu_{a})$ in order to minimize
\begin{equation}\label{eq:J}
J(\mu_{s}, \mu_{a})=\frac{1}{2}\,\sum_{i=1}^{n}\lp \frac{ L(x_{k_{i}}; \mu_{s}, \mu_{a})-m_{i}}{m_{i}}\rp^{2} ,
\end{equation}
where $(m_{i})_{i=1,\ldots, n}$ are measurements in $n$ different voxels centered at $(x_{k_{i}})_{i=1,\ldots, n}$.

The optimization method that we implemented to solve this problem is based on the Levenberg-Marquardt algorithm (see~\cite{Fletcher2013}). This gradient descent algorithm involves the computation of the gradient, as well as the Hessian matrix of the \emph{score function} $J$. It is described in pseudo-code in Algorithm~\ref{algo:LM}. In this description, we have set $H_{k}=\Hess(J)(\mu_{s}^{k},\mu_{a}^{k})$ for $k\geq 0$. The term $\Diag(H_{k})$ is the diagonal matrix of $H_{k}$, $\lambda_{k}$ is the so-called \emph{damping factor} which may be either constant or corrected at each step, and $\tau_{k}\in\R_{+}$ controls the step size of each iteration.

\begin{algorithm}
\begin{algorithmic}[1]
\REQUIRE measurements $(m_{i})_{i=1,\ldots, n}$, initial couple $(\mu_{s}^{0},\mu_{a}^{0})$, precision $\varepsilon>0$.
\STATE $k\leftarrow0$
\WHILE{$J(\mu_{s}^{k},\mu_{a}^{k})>\varepsilon$}
\STATE \label{algo:line:eq:LM}
$(\mu_{s}^{k+1},\mu_{a}^{k+1})\leftarrow(\mu_{s}^{k},\mu_{a}^{k})-\tau_{k}\,\lc H_{k}+\lambda_{k}\Diag(H_{k}) \rc^{-1} \nabla J(\mu_{s}^{k},\mu_{a}^{k})$
\STATE $k\leftarrow k+1$
\ENDWHILE
\ENSURE an approximation $(\mu_{s}^{k},\mu_{a}^{k})$ of the real parameters $(\mu_{s}^{*},\mu_{a}^{*})$
\end{algorithmic}\caption{Gradient descent algorithm for the estimation of $\mu_{a}$ and $\mu_{s}$}\label{algo:LM}
\end{algorithm}

The simple form of the objective function in~\eqref{eq:J} allows to express the term on the right hand side of line~\ref{algo:line:eq:LM} in Algorithm~\ref{algo:LM} explicitly as a function of the partial derivatives of $L$. Indeed, the gradient of $J$ is given by
\begin{equation}\label{eq:grad:J}
\nabla J(\cdot)=\sum_{i=1}^{n} \frac{ L(x_{k_i}; \cdot)-m_{i}}{m_{i}^{2}}\,\nabla L(x_{k_i}; \cdot),
\end{equation}
and its Hessian matrix is given by
\begin{equation}\label{eq:Hess:J}
\Hess(J)(\cdot)=\sum_{i=1}^{n} \left(\frac{ L(x_{k_i}; \cdot)-m_{i}}{m_{i}^{2}}\,\Hess(L)(x_{i};\cdot)+\frac{1}{m_{i}^{2}}\,\nabla L(x_{k_i}; \cdot)\nabla^{t} L(x_{k_i}; \cdot)\right).
\end{equation}

Moreover, as stated in the following proposition, the formal representation in Proposition~\ref{prop:expectation} allows to also use the Monte Carlo method MC-SOME in order to estimate the first order and the second order partial derivatives of $L$  which can be expressed similarly to~\eqref{eq:fluence:as:proba}.

\begin{Prop}\label{prop:L:derivatives}
The partial derivatives of $L(x_{k_i}; \mu_{s}, \mu_{a})$ can be expressed as the expectation of fully simulable random variables. Using the same notations as in~\eqref{eq:Lxk:RW}, they are given by
\begin{align}
\frac{\partial L}{\partial \mu_{a}} (x_{k_{i}}; \mu_{s}, \mu_{a})&=-\frac{c(1-\cos(\alpha))}{2\mu_a}\E\left( \ind_{\lac S_{N} \in V_{k_i}\rac} \sum_{j=0}^{N}R_j\right)
\label{eq:derivL:mua}
\\
{\rm and \quad } \frac{\partial L}{\partial \mu_{s}} (x_{k_{i}}; \mu_{s}, \mu_{a})&
=\frac{c(1-\cos(\alpha))}{2\mu_a}\E\left( \ind_{\lac S_{N} \in V_{k_i} \rac} \left(\frac{N}{\mu_s}-\sum_{j=0}^{N}R_j\right)\right).
\label{eq:derivL:mus}
\end{align}
\end{Prop}

\begin{proof}

We start by differentiating term-by-term the Neumann series of Corollary~\ref{cor:ERT:Neumann}. Let us first note that by definition of $L_i$ we have:
$$
\frac{\partial  L_{i} }{\partial \mu_{a}} (x,\omega_0;\mu_s,\mu_a)=-\int_{0}^{+\infty}r\exp(-(\mu_s+\mu_a) r )L_e(x-rw,w)dr.
$$ 
Moreover, for $n\geq1$, by definition of $T^{n}\circ T_i$ (see~\eqref{eq:def:Tn}), we also have:
\begin{multline*}
\frac{\partial \left[ T^n \circ T_i\right] L_e}{\partial \mu_{a}} (x,\omega_0;\mu_s,\mu_a)= \mu_s^n \int_{\R_{+}^{n}} dr_0\cdots dr_n \, \left( -\sum_{j=0}^{n} r_j\right)\exp\left( -(\mu_{s}+\mu_{a}) \sum_{j=0}^{n} r_j\right)
\\
\int_{(\Ss^{2})^{(n+1)}} d\sigma^{\otimes (n+1)}(\omega_0,\ldots,\omega_n) \prod_{j=0}^{n-1}f_{HG}\left( \omega_{j},\omega_{j+1}\right) \,
L_{e}\left( x - \sum_{k=0}^{n} r_k\omega_{k}  ,\omega_n\right).
\end{multline*}

Looking back at  Section~\ref{sec:proba:repr} and using the same notations, we deduce that 
$$
\sum_{n=0}^{\infty}\frac{\partial \left[ T^n \circ T_i\right] L_e}{\partial \mu_{a}} (x,\omega_0;\mu_s,\mu_a)= -\int_{\mathcal{A}} \nu(d\boldsymbol{r},d \boldsymbol{\omega})\,G_{x}(\boldsymbol{r}, \boldsymbol{\omega}) \sum_{i=0}^{|\boldsymbol{r}|} r_{i},$$
where we recall that $|\boldsymbol{r}|$ stands for the size of $\boldsymbol{r}$.
Assuming that the left-hand side coincides with the partial derivative $\frac{\partial L}{\partial \mu_{a}}$, then~\eqref{eq:derivL:mua} is found just like~\eqref{eq:Lxk:RW} and the same arguments provide~\eqref{eq:derivL:mus}, considering that
for all $n\ge 0$
$$\frac{\partial \left[ T^n \circ T_i\right] L_{e} }{\partial \mu_{s}}
=\frac{n}{\mu_{s}}\left[ T^n \circ T_i\right] L_{e}+\frac{\partial \left[ T^n \circ T_i\right] L_e}{\partial \mu_{a}}\,.
$$

To conclude the proof, notice that the match between the partial derivatives and the term-by-term differentiation of the Neumann series is ensured by the fact that the operator $T^n \circ T_i$ is infinitely continuously differentiable for all $n$ and by the uniform convergence of the corresponding sequences of truncated sums
$$s_{m}=\sum_{n=0}^{m} \frac{\partial [ T^n\circ T_i ]L_{e}}{\partial \mu_{s}} (x,\omega_0; \mu_s,\mu_a), \qquad m\geq 0.
$$
\end{proof}

\begin{Rem}\label{rem:second:derivatives} 
Similar formula to~\eqref{eq:derivL:mua} and~\eqref{eq:derivL:mus} can be easily found for the second order derivatives $\frac{\partial^{2} L}{\partial \mu_{s}^{2}}$, $\frac{\partial^{2} L}{\partial \mu_{s}\partial \mu_{a}}$ and $\frac{\partial^{2} L}{\partial \mu_{a}^{2}}$.
\end{Rem}

The probabilistic representation of $L(x_{k_i}; \mu_{s}, \mu_{a})$ in~\eqref{eq:fluence:as:proba} and its partial derivatives allows us to estimate the score $J(\mu_{s}, \mu_{a})$, its gradient and its Hessian matrix by Monte Carlo methods. A sole sample $(y_{1}, \ldots, y_{n})\in\mathcal{A}^{n}$ of $n$ observations of the random ray $Y$ can be used to estimate the expectations in $L$, $\nabla L(x_{k_i}; \cdot)$ and  $\Hess(L) (x_{i};\cdot)$ at the same time.  We denote these estimates by $\widehat L$, $\widehat\nabla L(x_{k_i}; \cdot)$ and  $\widehat\Hess(L) (x_{i};\cdot)$ and the corresponding score by $\widehat J$. The updating rule at line~\ref{algo:line:eq:LM} in Algorithm~\ref{algo:LM} becomes then
\begin{equation}\label{eq:LMalgo:estim}
(\mu_{s}^{k+1},\mu_{a}^{k+1})=(\mu_{s}^{k},\mu_{a}^{k})-\tau_{k}\,\lc \widehat H_{k}+\lambda_{k}\Diag(\widehat H_{k}) \rc^{-1} \widehat\nabla J(\mu_{s}^{k},\mu_{a}^{k}).
\end{equation}

\emph{Implementation and discussion.} The randomness coming from Monte Carlo estimation of the score $J$, of its gradient and of its Hessian matrix during the run of the algorithm, makes a precise estimation of the real values of $(\mu_{a}^{*}, \mu_{s}^{*})$ difficult. We observed that far from the real value $\mu_a^*$, the eigenvalues of the Hessian matrix $\hat H_{k}$ are very small and that their sign can vary a lot because of the volatility of the estimates. Conversely, near the real value $\mu_a^*$, the estimate $\hat H_{k}$ is more robust and its eigenvalues are almost always both positive, which legitimates the quadratic approximation of Levenberg-Marquardt algorithm. For these reasons, we implemented a \emph{hybrid} algorithm which chooses between the Levenberg-Marquardt descent and the classic steepest gradient descent depending on the sign of the eigenvalues of $H_{k}$. If they are both positive, one moves to the next point following~\eqref{eq:LMalgo:estim}, else one makes a move in the opposite direction of the gradient, $-\hat \nabla J$.

In Fig.~\ref{figure:descent:1} and~\ref{figure:descent:2}, we can see two examples of descent of our algorithm. The settings are the following: %
 \emph{(1)} we choose $n=3$ positions for the measurements: $x_{k_{1}}\in v_{2}$, $x_{k_{2}}\in v_{4}$, $x_{k_{3}}\in v_{6}$ (see Fig.~\ref{fig:comparison:choicevoxels}) and the values of the measurements $m_1,m_2$ and $m_3$ are taken from the database of simulations described in Sec.~\ref{sec:sensitivity} with the desired parameters, %
\emph{(2)} the anisotropy factor is set to $g=0.9$, %
\emph{(3)} the damping parameter is constant $\lambda=0.01$, %
\emph{(4)} the precision parameter (see Algo.~\ref{algo:LM}) is set to $\varepsilon=0.005$,  %
\emph{(5)} the sequence $(\tau_k)_{k\geq1}$ controlling the step size of each iteration is decreasing in $k$ and depends on the score of the iteration, as well as on the sign of the eigenvalues of $H_k$.

In Fig.~\ref{figure:descent:1}, the reference parameter are $(\mu_a^*=1, \mu_s^*=75)$. Notice the oscillations around $\mu_{a}^{*}$.  Those descent zigzags near the real value of $\mu_{a}$ are also apparent in the other descent in Fig.~\ref{figure:descent:2} for which $(\mu_a^*=1, \mu_s^*=105)$. They correspond to iteration where the descent is done according to the classic steepest descent. The iterations for which one moves more vertically correspond, as for them, to the case where the descent is done according to~\eqref{eq:LMalgo:estim}.  As we can see, a satisfying estimate of $\mu_{a}^{*}$ comes up rapidly, whereas $\mu_{s}^{*}$ is more difficult to approach. This is due to the low sensitivity of $J$ with respect to $\mu_s$ discussed in the previous section.

\begin{figure}[h!]
\begin{center}
\includegraphics[width=0.8\textwidth]{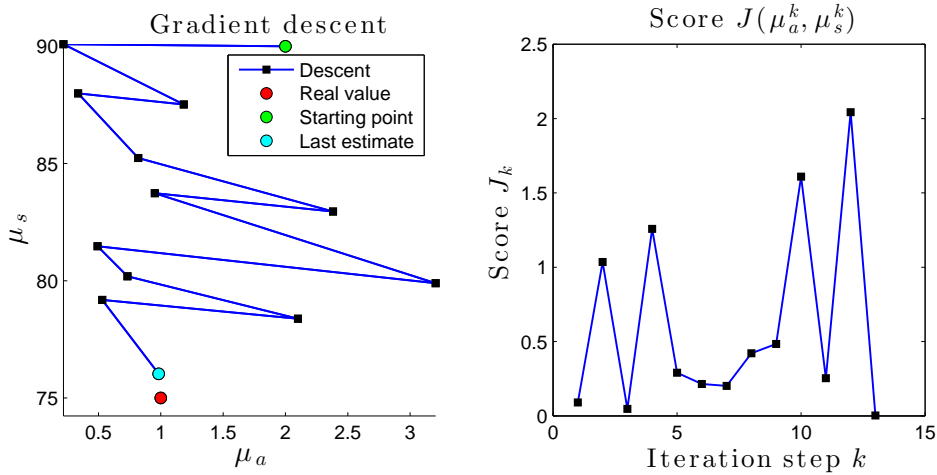}
\caption{Parameters estimation with an adaptation of Levenberg-Marquardt descent algorithm. The real value of $(\mu_{a}^{*}, \mu_{s}^{*})$ is $(1,75)$. The starting point is $(2,90)$. The final estimate is $(1.07, 76.02)$ with a score equal to $J=0.0035$.}
\label{figure:descent:1}
\end{center}
\end{figure}

\begin{figure}[h!]
\begin{center}
\includegraphics[width=0.8\textwidth]{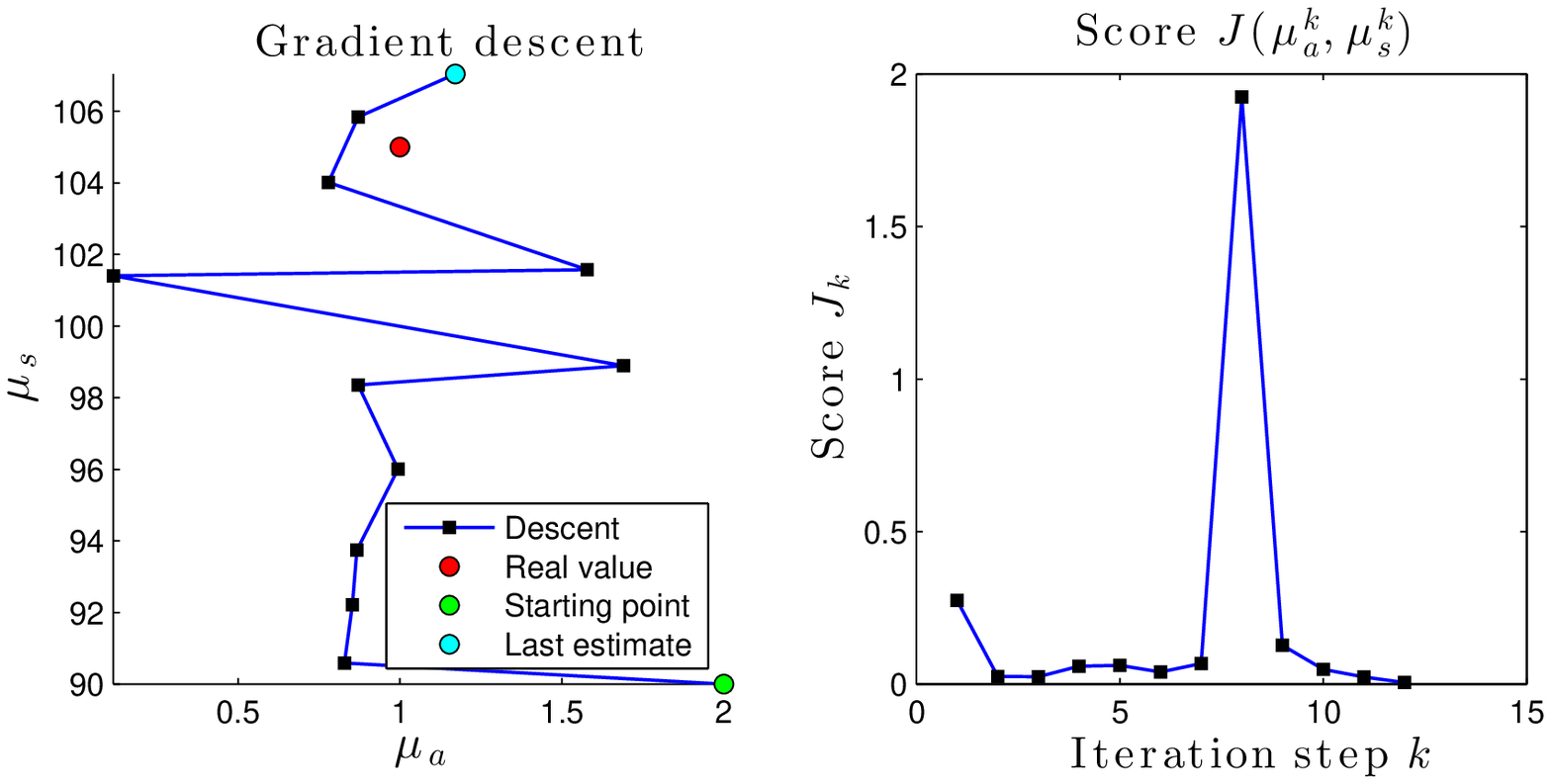}
\caption{Parameters estimation with an adaptation of Levenberg-Marquardt descent algorithm. The real value of $(\mu_{a}^{*}, \mu_{s}^{*})$ is $(1,105)$. The starting point is $(2,90)$. The final estimate is $(1.17, 107.04)$ with a score equal to $J=0.0048$.}
\label{figure:descent:2}
\end{center}
\end{figure}

\FloatBarrier

\addcontentsline{toc}{section}{References}
\small
\bibliographystyle{abbrv}
\bibliography{biblio_article_1}

\end{document}